\documentclass[12pt,oneside,reqno]{amsart}
\usepackage[latin9]{inputenc}
\usepackage{geometry}
\usepackage{mathrsfs}
\usepackage{mathtools}
\usepackage{amstext}
\usepackage{amsthm}
\usepackage{amssymb}
\usepackage{mathdots}
\usepackage{graphicx}
\usepackage[unicode=true,pdfusetitle,
 bookmarks=true,bookmarksnumbered=true,bookmarksopen=true,bookmarksopenlevel=1,
 breaklinks=false,pdfborder={0 0 1},backref=section,colorlinks=false]
 {hyperref}
\hypersetup{
 hypertex,bookmarksnumbered=true}

\makeatletter
\numberwithin{equation}{section}
\numberwithin{figure}{section}
\theoremstyle{plain}
\newtheorem{thm}{\protect\theoremname}
\theoremstyle{definition}
\newtheorem{defn}[thm]{\protect\definitionname}
\theoremstyle{remark}
\newtheorem{rem}[thm]{\protect\remarkname}
\theoremstyle{plain}
\newtheorem{cor}[thm]{\protect\corollaryname}

\usepackage{chngcntr}
\counterwithout{figure}{section}

\makeatother

\providecommand{\corollaryname}{Corollary}
\providecommand{\definitionname}{Definition}
\providecommand{\remarkname}{Remark}
\providecommand{\theoremname}{Theorem}

\begin{document}
\title{Exceptional points of degeneracy in traveling wave tubes}
\author{Alexander Figotin}
\address{University of California at Irvine, CA 92967}
\begin{abstract}
Traveling wave tube (TWT) is a powerful vacuum electronic device used
to amplify radio-frequency (RF) signals with numerous applications,
including radar, television and telephone satellite communications.
TWT design in a nutshell comprises of a pencil-like electron beam
(e-beam) in vacuum interacting with guiding it slow-wave structure
(SWS). In our studies here the e-beam is represented by one-dimensional
electron flow and SWS is represented by a transmission line (TL).
The interaction between the e-beam and the TL is modeled by an analytic
theory that generalizes the well-known Pierce model by taking into
account the so-called space-charge effects particularly electron-to-electron
repulsion (debunching). Many important aspects of the analytic theory
of TWTs have been already analyzed in our monograph on the subject.
The focus of the studies here is on degeneracies of the TWT dispersion
relations particularly on exceptional points of degeneracy and their
applications. The term exceptional point of degeneracy (EPD) refers
to the property of the relevant matrix to have nontrivial Jordan block
structure. Using special parameterization particularly suited to chosen
EPD we derive exact formulas for the relevant Jordan basis including
the eigenvectors and the so-called root vector associated with the
Jordan block. Based on these studies we develop constructive approach
to sensing of small signals.
\end{abstract}

\maketitle

\section{Introduction\label{sec:int-twtj}}

There is growing interest to electromagnetic system exhibiting Jordan
eigenvector degeneracy, which is a degeneracy of the system evolution
matrix when not only some eigenvalues coincide but the corresponding
eigenvectors coincide also. The degeneracy of this type is sometimes
referred to as exceptional point of degeneracy (EPD), \cite[II.1]{Kato}.
A particularly important class of applications of EPDs is sensing,
\cite{CheN}. \cite{PeLiXu}, \cite{Wie}, \cite{Wie1}, \cite{KNAC},
\cite{OGC}.

In our prior work in \cite{FigSynbJ}, \cite{FigSynbJ,FigPert} we
advanced and studied simple circuits exhibiting EPDs and their applications
to sensing. Operation of electric circuits though is limited to frequencies
up to hundreds of MHz, and to overcome this limitation other physical
systems that can operate at higher frequencies must be considered.
Our prior studies of traveling wave tubes (TWT) in \cite[4, 7, 13, 14, 54, 55]{FigTWTbk}
demonstrate that TWTs always have EPDs. Operating frequencies of TWTs
can go up to hundreds of GHz and even into THz frequency range, \cite{BoosVE},
\cite{Burt} and for this reason they are the primary subject of our
studies here. For more applications of EPDs for traveling wave tubes
see \cite{OTC}, \cite{OVFC}, \cite{OVFC1}, \cite{VOFC}.

We start with a concise review of the basics of traveling wave tubes.
Traveling wave tube (TWT) utilizes the energy of the electron beam
(e-beam) as a flow of free electrons in a vacuum and converts it into
an RF signal, see Fig. \ref{fig:TWT1}. To facilitate energy conversion
and signal amplification, the electron beam is enclosed in the so-called
\emph{slow wave structure} (SWS), which supports waves that are slow
enough to effectively interact with the e-beam. As a result of this
interaction, the kinetic energy of electrons is converted into the
electromagnetic energy stored in the field, \cite{Gilm1}, \cite{Tsim},
\cite[2.2]{Nusi}, \cite[4]{SchaB}. Consequently, the \emph{key operational
principle of a TWT is a positive feedback interaction between the
slow-wave structure and the flow of electrons}. The physical mechanism
of radiation generation and its amplification is the electron bunching
caused by the acceleration and deceleration of electrons along the
e-beam.
\begin{figure}[h]
\centering{}\includegraphics[scale=0.5]{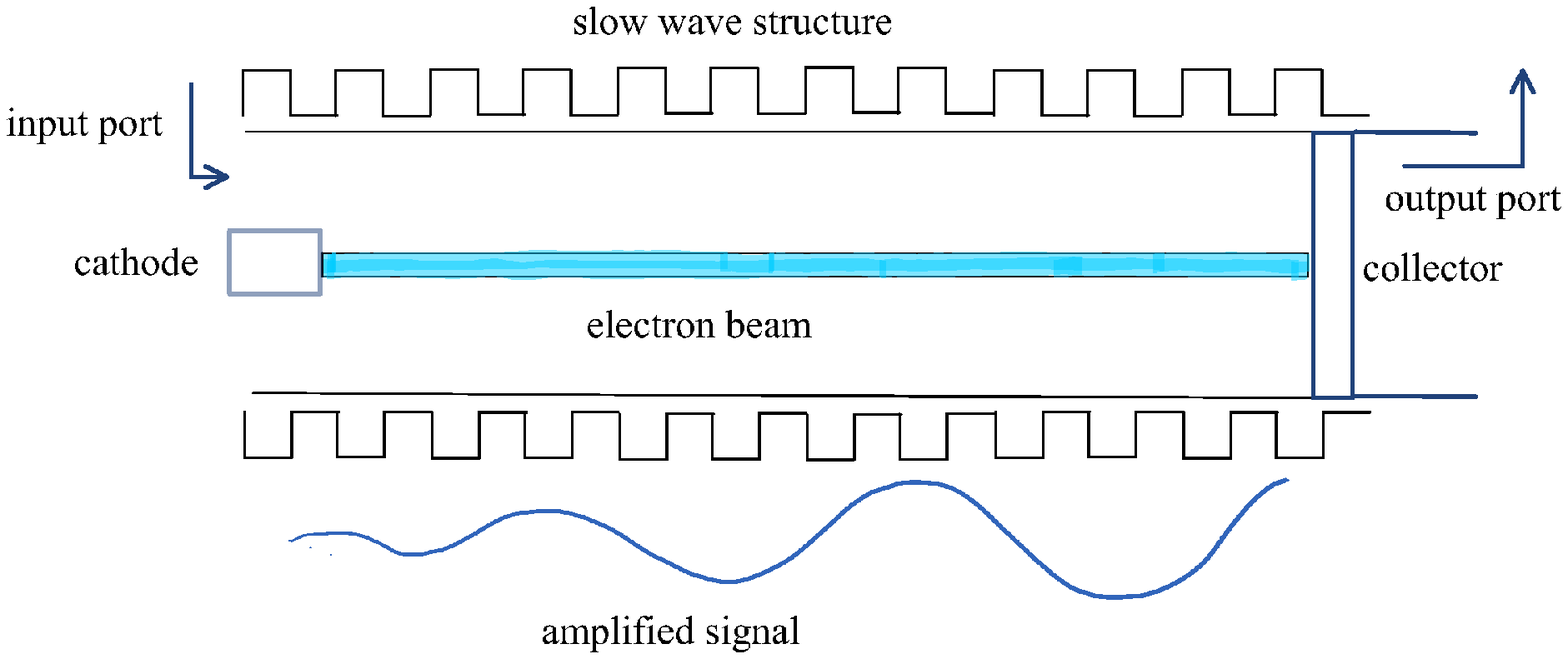}\caption{\label{fig:TWT1} The upper picture is a schematic presentation of
a traveling wave tube. The lower picture shows an RF perturbation
in the form of a space-charge wave that gets amplified exponentially
as it propagates through the traveling wave tube.}
\end{figure}

A schematic sketch of typical TWT is shown in Fig. \ref{fig:TWT1}.
Such a typical TWT consists of a vacuum tube containing an e-beam
that passes down the middle of an SWS such as an RF circuit. It operates
as follows. The left end of the RF circuit is fed with a low-powered
RF signal to be amplified. The SWS electromagnetic field acts upon
the e-beam causing electron bunching and the formation of the so-called
\emph{space-charge wave}. In turn, the electromagnetic field generated
by the space charge wave induces more current back into the RF circuit
with a consequent enhancement of electron bunching. As a result, the
EM field is amplified as the RF signal passes down the structure until
a saturation regime is reached and a large RF signal is collected
at the output. The role of the SWS is to provide slow-wave modes to
match up with the velocity of the electrons in the e-beam. This velocity
is usually a small fraction of the speed of light. Importantly, synchronism
is required for effective in-phase interaction between the SWS and
the e-beam with optimal extraction of the kinetic energy of the electrons.
A typical simple SWS is the helix, which reduces the speed of propagation
according to its pitch. The TWT is designed so that the RF signal
travels along the tube at nearly the same speed as electrons in the
e-beam to facilitate effective coupling. Technical details on the
designs and operation of TWTs can be found in \cite{Gilm1}, \cite[4]{Nusi}
\cite{PierTWT}, \cite{Tsim}. As for a rich and interesting history
of traveling wave tubes, we refer the reader to \cite{MAEAD} and
references therein.

An effective mathematical model for a TWT interacting with the e-beam
was introduced by Pierce \cite[I]{Pier51}, \cite{PierTWT}. The Pierce
model is one-dimensional; it accounts for the wave amplification,
energy extraction from the e-beam and its conversion into microwave
radiation in the TWT \cite{Gilm1}, \cite{Gilm}, \cite[4]{Nusi},
\cite[4]{SchaB}, \cite{Tsim}. This model captures remarkably well
significant features of the wave amplification and the beam-wave energy
transfer, and is still used for basic design estimates. In our paper
\cite{FigRey1}, we have constructed a Lagrangian field theory by
generalizing and extending the Pierce theory to the case of a possibly
inhomogeneous MTL coupled to the e-beam. This work was extended to
an analytic theory of multi-stream electron beams in traveling wave
tubes in \cite{FigTWTbk}.

According to our analytic theory the TWT dispersion relations always
have EPDs which can be effectively found \cite[4, 7, 13, 14, 54, 55]{FigTWTbk}.
We study here the simplest TWT model that generalizes the Pierce model
by integrating into it the space-charge effects. We introduce for
this model a special parameterization that allow for (i) explicit
representation of chosen EPD associated with a Jordan block of size
2; (ii) exact formulas for the Jordan basis including the eigenvectors
and the so-called root vector associated with the Jordan block.

The paper is organized as follows. In Section \ref{sec:twt-mod} we
review our analytic model of TWT introduced and studied in \cite[4, 24]{FigTWTbk}.
In Section \ref{sec:twt-epd} we carry out detailed studies of the
relevant TWT matrix and its Jordan form. In Section \ref{sec:sensing}
we (i) carry out the perturbation analysis of the relevant TWT matrix;
(ii) develop constructive approach for using the TWT EPD for sensing
of small signals.

While quoting monographs we identify the relevant sections as follows.
Reference {[}X,Y{]} refers to Section/Chapter ``Y'' of monograph
(article) ``X'', whereas {[}X, p. Y{]} refers to page ``Y'' of
monograph (article) ``X''. For instance, reference {[}2, VI.3{]}
refers to monograph {[}2{]}, Section VI.3; reference {[}2, p. 131{]}
refers to page 131 of monograph {[}2{]}.

\section{An analytic model of the traveling wave tube\label{sec:twt-mod}}

We concisely review here an analytic model of the traveling wave tube
introduced and studied in our monograph \cite[4, 24]{FigTWTbk}. According
to this model an ideal TWT is represented by a single-stream e-beam
interacting with single transmission line. This model is a generalization
of the Pierce model \cite[I]{Pier51}, \cite{PierTWT} and its parameters
are as follows. The main parameter describing the single-stream e-beam
is e-beam coefficient
\begin{equation}
\beta=\frac{\sigma_{\mathrm{B}}}{4\pi}R_{\mathrm{sc}}^{2}\omega_{\mathrm{p}}^{2}=\frac{e^{2}}{m}R_{\mathrm{sc}}^{2}\sigma_{\mathrm{B}}\mathring{n},\quad\omega_{\mathrm{p}}^{2}=\frac{4\pi\mathring{n}e^{2}}{m},\label{eq:T1B1beta1aj}
\end{equation}
where $-e$ is electron charge with $e>0$, $m$ is the electron mass,
$\omega_{\mathrm{p}}$ is the e-beam plasma frequency, $\sigma_{\mathrm{B}}$
is the area of the cross-section of the e-beam, the constant $R_{\mathrm{sc}}$
is the\emph{ }plasma frequency reduction factor that accounts phenomenologically
for finite dimensions of the e-beam cylinder as well as geometric
features of the slow-wave structure, \cite[41, 63]{FigTWTbk}, and
$\mathring{n}$ is the density of the number of electrons. The single-stream
e-beam has steady velocity $\mathring{v}>0$.

As for the single transmission line, its shunt capacitance per unit
of length is a real number $C>0$ and its inductance per unit of length
is another real number $L>0$. The coupling constant $0<b\leq1$ is
a number also, see \cite[3]{FigTWTbk} for more details. The TL single
characteristic velocity $w$ and the single \emph{TL principal coefficient}
$\theta$ defined by
\begin{equation}
w=\frac{1}{\sqrt{CL}},\quad\theta=\frac{b^{2}}{C}.\label{eq:T1B1beta1bj}
\end{equation}
Following to \cite[3]{FigTWTbk} we assume that
\begin{equation}
0<\mathring{v}<w.\label{eq:T1B1beta1cj}
\end{equation}

\subsection{TWT system Lagrangian and evolution equations\label{subsec:two-lag-ev}}

Following to developments in \cite{FigTWTbk} we introduce the \emph{TWT
principal parameter} $\bar{\gamma}=\theta\beta$. This parameter in
view of equations (\ref{eq:T1B1beta1aj}) and (\ref{eq:T1B1beta1bj})
can be represented as follows
\begin{equation}
\gamma=\theta\beta=\frac{b^{2}}{C}\frac{\sigma_{\mathrm{B}}}{4\pi}R_{\mathrm{sc}}^{2}\omega_{\mathrm{p}}^{2}=\frac{b^{2}}{C}\frac{e^{2}}{m}R_{\mathrm{sc}}^{2}\sigma_{\mathrm{B}}\mathring{n},\quad\theta=\frac{b^{2}}{C},\quad\beta=\frac{e^{2}}{m}R_{\mathrm{sc}}^{2}\sigma_{\mathrm{B}}\mathring{n}.\label{eq:DsT1B1hj}
\end{equation}
The TWT-system Lagrangian $\mathcal{L}{}_{\mathrm{TB}}$ is defined
by \cite[4, 24]{FigTWTbk}
\begin{gather}
\mathcal{L}{}_{\mathrm{TB}}=\mathcal{L}_{\mathrm{Tb}}+\mathcal{L}_{\mathrm{B}},\label{eq:T1B1beta1dj}\\
\mathcal{L}_{\mathrm{Tb}}=\frac{L}{2}\left(\partial_{t}Q\right)^{2}-\frac{1}{2C}\left(\partial_{z}Q+b\partial_{z}q\right)^{2},\;\mathcal{L}_{\mathrm{B}}=\frac{1}{2\beta}\left(\partial_{t}q+\mathring{v}\partial_{z}q\right)^{2}-\frac{2\pi}{\sigma_{\mathrm{B}}}q^{2},\nonumber 
\end{gather}
where $q\left(z,t\right)$ and $Q\left(z,t\right)$\emph{ are charges}
associated with the e-beam and the TL defined as time integrals of
the corresponding e-beam currents $i(z,t)$ and TL current $I(z,t)$,
that is
\begin{equation}
q(z,t)=\int^{t}i(z,t^{\prime})\,\mathrm{d}t^{\prime},\quad.Q(z,t)=\int^{t}I(z,t^{\prime})\,\mathrm{d}t^{\prime}.\label{eq:MTLQ1aj}
\end{equation}
The corresponding Euler-Lagrange equations are the following system
of second-order differential equations
\begin{gather}
L\partial_{t}^{2}Q-\partial_{z}\left[C^{-1}\left(\partial_{z}Q+b\partial_{z}q\right)\right]=0,\label{eq:T1B1beta2aj}\\
\frac{1}{\beta}\left(\partial_{t}+\mathring{v}\partial_{z}\right)^{2}q+\frac{4\pi}{\sigma_{\mathrm{B}}}q-b\partial_{z}\left[C^{-1}\left(\partial_{z}Q+b\partial_{z}q\right)\right]=0,\quad\beta=\frac{\sigma_{\mathrm{B}}}{4\pi}R_{\mathrm{sc}}^{2}\omega_{\mathrm{p}}^{2}.\label{eq:T1B1beta2bj}
\end{gather}
The Fourier transformation (see Section \ref{sec:four}) in time $t$
and space variable $z$ of equations (\ref{eq:T1B1beta2aj}) and (\ref{eq:T1B1beta2bj})
yields
\begin{gather}
\left(k^{2}C^{-1}-\omega^{2}L\right)\hat{Q}+k^{2}C^{-1}b\hat{q}=0,\label{eq:T1B1beta2cj}\\
\frac{4\pi}{\sigma_{\mathrm{B}}}\left[1-\frac{\left(\omega-\mathring{v}k\right)^{2}}{R_{\mathrm{sc}}^{2}\omega_{\mathrm{p}}^{2}}\right]\hat{q}+k^{2}bC^{-1}\left[b\hat{q}+\hat{Q}\right]=0,\label{eq:T1B1beta2dj}
\end{gather}
where $\omega$ and $k=k\left(\omega\right)$ are the frequency and
the wavenumber, respectively, and functions $\hat{Q}=\hat{Q}\left(\omega,k\right)$
and $\hat{q}=\hat{q}\left(\omega,k\right)$ are the Fourier transforms
of the system vector variables $Q\left(t,z\right)$ and $q\left(t,z\right)$,
see Appendix \ref{sec:four}. In this case, the general TWT-system
eigenmodes are of the form
\begin{equation}
Q\left(z,t\right)=\hat{Q}\left(k,\omega\right)\mathrm{e}^{-\mathrm{i}\left(\omega t-kz\right)},\quad q\left(z,t\right)=\hat{q}\left(k,\omega\right)\mathrm{e}^{-\mathrm{i}\left(\omega t-kz\right)},\label{eq:T1B1qQ1aj}
\end{equation}
where as it turns out
\begin{equation}
\hat{q}\left(k,\omega\right)=a_{0}\frac{1}{\left(u-\mathring{v}\right)^{2}},\quad\hat{Q}\left(k,\omega\right)=a_{0}\frac{bw^{2}}{\left(u-\mathring{v}\right)^{2}\left(u^{2}-w^{2}\right)},\quad u=\frac{\omega}{k}.\label{eq:T1B1qQ1bj}
\end{equation}
In equation (\ref{eq:T1B1qQ1bj}) $a_{0}$ is a constant of proper
physical dimensions and, importantly, $u$ is the complex-valued characteristic
velocity satisfying characteristic equation (\ref{eq:DsT1B1cj}) provided
below.

Notice that dividing equations (\ref{eq:T1B1beta2cj})-(\ref{eq:T1B1beta2dj})
by $k^{2}$ yields
\begin{gather}
\left(k^{2}C^{-1}-u^{2}L\right)\hat{Q}+C^{-1}b\hat{q}=0,\quad u=\frac{\omega}{k},\label{eq:T1B1beta2cuj}\\
\frac{4\pi}{\sigma_{\mathrm{B}}}\left[\frac{u^{2}}{\omega^{2}}-\frac{\left(u-\mathring{v}\right)^{2}}{\Omega_{\mathrm{p}}^{2}}\right]\hat{q}+bC^{-1}\left[b\hat{q}+\hat{Q}\right]=0,\quad\Omega_{\mathrm{p}}=R_{\mathrm{sc}}\omega_{\mathrm{p}}.\label{eq:T1B1beta2duj}
\end{gather}
Concise matrix form of equations (\ref{eq:T1B1beta2cuj})-(\ref{eq:T1B1beta2duj})
is the following eigenvalue type problem for $k$ and $x$ assuming
that $\omega$ is fixed
\begin{equation}
M_{k\omega}x=0,\quad M_{k\omega}=\left[\begin{array}{rr}
-L\omega^{2}+\frac{k^{2}}{C} & \frac{bk^{2}}{C}\\
\frac{bk^{2}}{C} & \frac{b^{2}k^{2}}{C}+\frac{4\pi}{\sigma}\left(1-\frac{\left(k\mathring{v}-\omega\right)^{2}}{\Omega_{\mathrm{p}}^{2}}\right)
\end{array}\right],\quad x=\left[\begin{array}{r}
\hat{Q}\\
\hat{q}
\end{array}\right],\label{eq:MLkC1a}
\end{equation}
The problem (\ref{eq:MLkC1a}) is equivalent to another eigenvalue
type problem for $u$ and $x$ assuming that $\omega$ is fixed
\begin{gather}
M_{u\omega}x=0,\quad M_{u\omega}=\left[\begin{array}{rr}
\frac{1}{u^{2}}-\frac{1}{w^{2}} & \frac{b}{u^{2}}\\
\frac{b}{u^{2}} & \frac{b^{2}}{\gamma\breve{\omega}^{2}}+\frac{b^{2}}{u^{2}}-\frac{b^{2}\left(u-\mathring{v}\right)^{2}}{\gamma u^{2}}
\end{array}\right],\quad x=\left[\begin{array}{r}
\hat{Q}\\
\hat{q}
\end{array}\right],\label{eq:MLkC1b}\\
u=\frac{\omega}{k},\quad\breve{\omega}=\frac{\omega}{\Omega_{\mathrm{p}}}.\nonumber 
\end{gather}
where matrix $M_{u\omega}$ encodes the information about TWT eigenmodes
and we refer to it as \emph{TWT principal matrix}. As we show below
this eigenvalue problem can be recast in terms of the theory of matrix
polynomials reviewed in Section \ref{sec:mat-poly}.

Notice the first equation in (\ref{eq:MLkC1b}) is also kind of eigenvalue
problem where being given $w$, $\mathring{v}$, $b$, $\gamma$ and
$\breve{\omega}$ we need to find $u$ and two-dimensional nonzero
vector $x$ that solve it. The problem of finding such $u$ is evidently
reduced evidently to the characteristic equation
\begin{equation}
\det\left\{ M_{u\omega}\right\} =0,\label{eq:MLkC1c}
\end{equation}
where $M_{u\omega}$ is TWT principal matrix defined in equations
(\ref{eq:MLkC1b}). After elementary algebraic transformations of
equation (\ref{eq:MLkC1c}) we find that it is equivalent to the following
equation
\begin{equation}
\mathscr{D}\left(u,\gamma\right)=\frac{\gamma}{w^{2}-u^{2}}+\frac{\left(u-\mathring{v}\right)^{2}}{u^{2}}=\frac{1}{\breve{\omega}^{2}},\quad\breve{\omega}=\frac{\omega}{\Omega_{\mathrm{p}}},\quad u=\frac{\omega}{k},\label{eq:DsT1B1cj}
\end{equation}
and we refer to function $\mathscr{D}\left(u,\gamma\right)$ as the
\emph{characteristic function}. We refer to solutions of the characteristic
equation (\ref{eq:DsT1B1cj}) as \emph{characteristic velocities}.
Since the dimensionless frequency $\breve{\omega}$ in characteristic
equation (\ref{eq:DsT1B1cj}) is real, the characteristic equation
is equivalent to the following system of equations:
\begin{gather}
\Im\left\{ \mathscr{D}\left(u,\bar{\gamma}\right)\right\} =0,\quad\Re\left\{ \mathscr{D}\left(u,\bar{\gamma}\right)\right\} =\frac{1}{\breve{\omega}^{2}}\geq0.\label{eq:MTLBdis2aarj}
\end{gather}
With the above equations in mind, we denote the set of all characteristic
velocities $u$ by $\mathscr{U}_{\mathrm{TB}}^{+}=\mathscr{U}_{\mathrm{TB}}^{+}\left(\bar{\gamma}\right)$
\begin{equation}
\mathscr{U}_{\mathrm{TB}}^{+}=\mathscr{U}_{\mathrm{TB}}^{+}\left(\bar{\gamma}\right)=\left\{ u:\Re\left\{ \mathscr{D}\left(u,\bar{\gamma}\right)\right\} \geq0\right\} \label{eq:UTBplus}
\end{equation}
with superscript ``+'' being a reminder of $\Re\left\{ \mathscr{D}\left(u,\bar{\gamma}\right)\right\} \geq0$.

The TWT characteristic function $\mathscr{D}\left(u,\bar{\gamma}\right)$
defined by equation (\ref{eq:T1B1beta2cj}) can be represented in
the form 
\begin{equation}
\mathscr{D}\left(u,\gamma\right)=\gamma\mathscr{D}_{\mathrm{T}}\left(u\right)+\mathscr{D}_{\mathrm{B}}\left(u\right)=\frac{\gamma}{w^{2}-u^{2}}+\frac{\left(u-\mathring{v}\right)^{2}}{u^{2}},\quad0<\mathring{v}<w,\quad\gamma=\theta\beta,\label{eq:DsT1B1aj}
\end{equation}
where $\mathscr{D}_{\mathrm{T}}\left(u\right)$ and $\mathscr{D}_{\mathrm{B}}\left(u\right)$
are respectively characteristic functions of the TL and e-beam defined
by
\begin{equation}
\mathscr{D}_{\mathrm{T}}\left(u\right)=u^{-2}\Delta_{\mathrm{T}}\left(u\right)=\frac{1}{w^{2}-u^{2}},\quad\mathscr{D}_{\mathrm{B}}\left(u\right)=u^{-2}\Delta_{\mathrm{B}}^{-1}\left(u\right)=\frac{\left(u-\mathring{v}\right)^{2}}{u^{2}}.\label{eq:DsT1B1bj}
\end{equation}
Functions $\Delta_{\mathrm{T}}\left(u\right)$ and $\Delta_{\mathrm{B}}\left(u\right)$
are defined in turn by
\begin{equation}
\Delta_{\mathrm{T}}\left(u\right)=\frac{u^{2}}{w^{2}-u^{2}},\quad\Delta_{\mathrm{B}}\left(u\right)=\frac{1}{\left(u-\mathring{v}\right)^{2}},\quad u=\frac{\omega}{k}.\label{eq:DsT1B1baj}
\end{equation}

Using dimensionless variables
\begin{equation}
\breve{\omega}=\frac{\omega}{\Omega_{\mathrm{p}}},\quad\Omega_{\mathrm{p}}=R_{\mathrm{sc}}\omega_{\mathrm{p}},\quad\check{\gamma}=\frac{\gamma}{\mathring{v}{}^{2}},\quad\check{u}=\frac{u}{\mathring{v}}=\frac{\omega}{k\mathring{v}},\quad\chi=\frac{w}{\mathring{v}}.\label{eq:MLkC1d}
\end{equation}
 we can recast eigenmode equation (\ref{eq:MLkC1b}) as follows
\begin{equation}
M_{u\omega}x=0,\quad M_{u\omega}=\left[\begin{array}{rr}
\frac{1}{\check{u}{}^{2}}-\frac{1}{\chi^{2}} & \frac{b}{\check{u}{}^{2}}\\
\frac{b}{\check{u}{}^{2}} & \left(\frac{1}{\check{u}{}^{2}}+\frac{1}{\check{\gamma}}\left(\frac{1}{\check{\omega}{}^{2}}-\frac{\left(\check{u}-1\right)^{2}}{\check{u}^{2}}\right)\right)b^{2}
\end{array}\right],\quad x=\left[\begin{array}{r}
\hat{Q}\\
\hat{q}
\end{array}\right].\label{eq:MLkC1e}
\end{equation}
It is convenient to recast the first equation into the standard matrix
polynomial in $\check{u}$ by multiplying it by $\check{u}{}^{2}$.
That action after elementary algebraic transformations yields
\begin{gather}
\mathsf{M}_{u\omega}x=0,\quad\mathsf{M}_{u\omega}=D_{b}^{\prime}\left[\begin{array}{rr}
\check{u}{}^{2}-\chi^{2} & -\chi^{2}\\
\frac{\check{\gamma}}{\check{\omega}{}^{-2}-1} & \check{u}{}^{2}+\frac{2\check{u}}{\check{\omega}{}^{-2}-1}+\frac{\check{\gamma}-1}{\check{\omega}{}^{-2}-1}
\end{array}\right]D_{b},\quad x=\left[\begin{array}{r}
\hat{Q}\\
\hat{q}
\end{array}\right],\label{eq:MLkC2a}
\end{gather}
where
\begin{equation}
D_{b}^{\prime}=\left[\begin{array}{rr}
-\frac{1}{\chi^{2}} & 0\\
0 & \frac{b\left(\check{\omega}{}^{-2}-1\right)}{\check{\gamma}}
\end{array}\right],\quad D_{b}=\left[\begin{array}{rr}
1 & 0\\
0 & b
\end{array}\right].\label{eq:MLkC2b}
\end{equation}
Matrix polynomial equation (\ref{eq:MLkC2a}) can in turn be readily
recast into the following equivalent form
\begin{gather}
\mathsf{M}_{u\omega}x=0,\quad\mathsf{M}_{u\omega}=\left[\begin{array}{rr}
\check{u}{}^{2}-\chi^{2} & -\chi^{2}\\
\frac{\check{\gamma}}{\check{\omega}{}^{-2}-1} & \check{u}{}^{2}+\frac{2\check{u}}{\check{\omega}{}^{-2}-1}+\frac{\check{\gamma}-1}{\check{\omega}{}^{-2}-1}
\end{array}\right],\quad x=\left[\begin{array}{r}
\hat{Q}\\
b\hat{q}
\end{array}\right].\label{eq:MLkC2c}
\end{gather}
We refer to matrix $\mathsf{M}_{u\omega}$ defined in equations (\ref{eq:MLkC2c})
\emph{TWT matrix polynomial}. The theory of matrix polynomials is
reviewed in Section \ref{sec:mat-poly}. From now on the matrix monic
polynomial in $\check{u}$ equation (\ref{eq:MLkC2c}) will be our
preferred form. The corresponding characteristic equation equivalent
to $\det\left\{ \mathsf{M}_{u\omega}\right\} =0$ and related to it
characteristic function $\mathscr{D}\left(\check{u},\check{\gamma}\right)$
and the corresponding equation are as follows

\begin{gather}
\mathscr{D}\left(\check{u},\check{\gamma}\right)=\frac{\check{\gamma}}{\chi^{2}-\check{u}{}^{2}}+\frac{\left(\check{u}-1\right)^{2}}{\check{u}{}^{2}}=\frac{1}{\breve{\omega}^{2}}.\label{eq:DsT1B1dj}
\end{gather}

The conventional plot of the characteristic function $\mathscr{D}\left(u,\gamma\right)$
helps to visualize solutions $u$ to the characteristic equations
(\ref{eq:DsT1B1cj}), that is, $\mathscr{D}\left(u,\bar{\gamma}\right)=\frac{1}{\breve{\omega}^{2}}$,
when the solutions are real-valued. But a non-real solution $u$ to
the characteristic equation, that is when $\Im\left\{ u\right\} \neq0$
is ``invisible'' in the conventional plot of the characteristic
function $\mathscr{D}\left(u,\bar{\gamma}\right)$. To remedy this,
we extend the conventional following to \cite[8]{FigTWTbk}. Suppose
that for $D>0$ equation $\mathscr{D}\left(u,\bar{\gamma}\right)=D$
has a complex-valued solution $u$, $\Im\left\{ u\right\} \neq0$.
We then represent such complex $u$ in the plot of $\mathscr{D}\left(u,\bar{\gamma}\right)$
as a point

\begin{equation}
\left(\Re\left\{ u\right\} ,D\right),\text{ where }\mathscr{D}\left(u,\bar{\gamma}\right)=D>0.\label{eq:topinstab1aj}
\end{equation}
We distinguish such points $\left(\Re\left\{ u\right\} ,D\right)$
by plotting curves comprised from them as solid (brown). We refer
to the so extended plot as the \emph{plot of the characteristic function
with instability branches}. We also refer to the endpoints $\left(\mathsf{u}_{0},\mathscr{D}\left(\mathsf{u}_{0},\bar{\gamma}\right)\right)$
of the instability branches as the \emph{characteristic function instability
nodes}. When generating plots and numerical examples we often use
the following data
\begin{equation}
\mathring{v}=1,\quad w=1.1,\quad\gamma=3.\label{eq:datavwgam}
\end{equation}

Figures \ref{fig:char-instab-T1B1-gam3-2}, \ref{fig:char-instab-T1B1-gam3-1},
and \ref{fig:char-instab-T1B1-gam3-3} show plots of the characteristic
function $\mathscr{D}\left(u,\gamma\right)$ defined in equations
(\ref{eq:DsT1B1cj}) with its instability branches.
\begin{figure}
\begin{centering}
\includegraphics[scale=0.55]{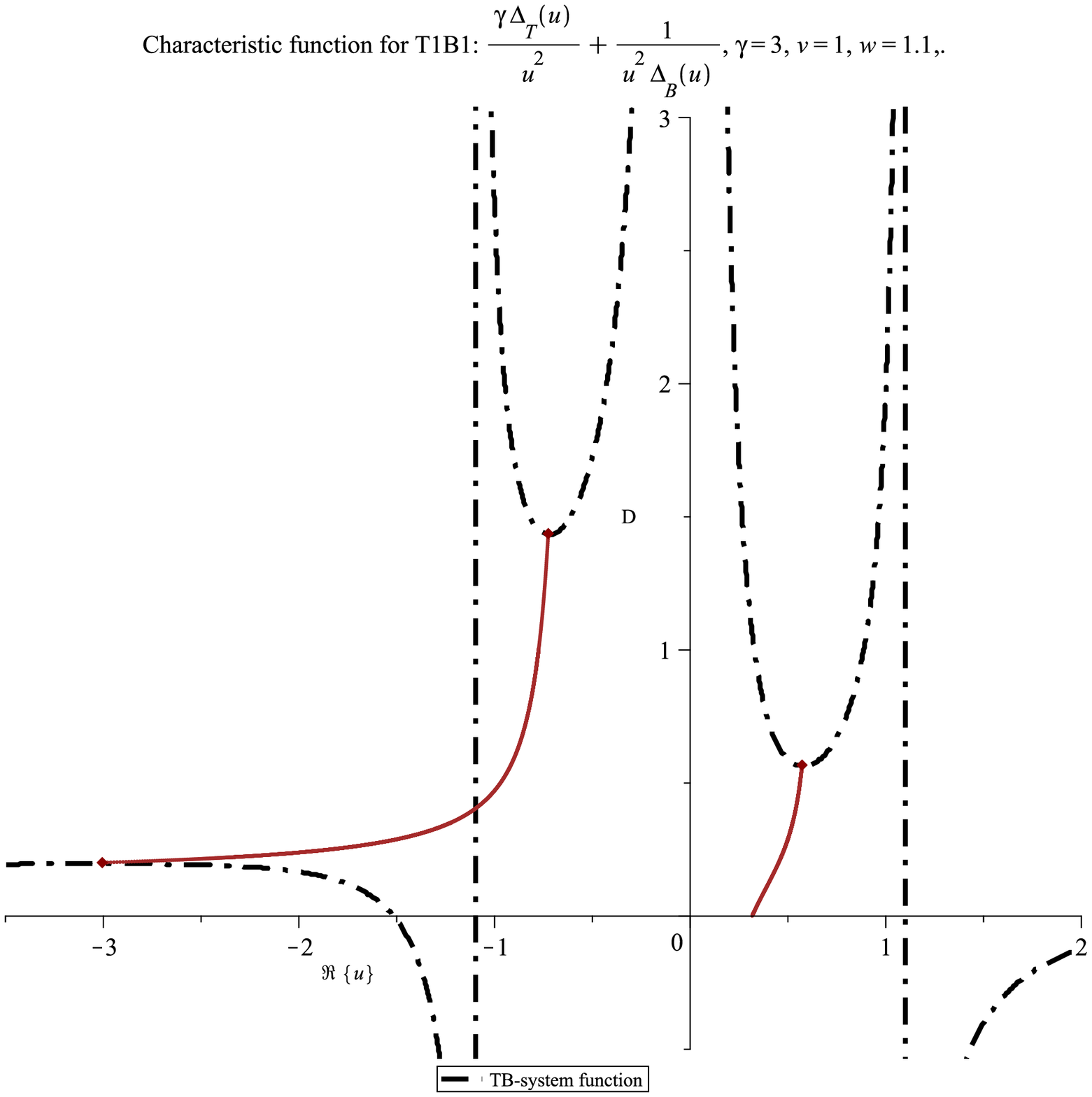}
\par\end{centering}
\centering{}\caption{\label{fig:char-instab-T1B1-gam3-2} Single-stream e-beam coupled
to a single TL. Plot of the characteristic function $\mathscr{D}\left(u,\gamma\right)=\frac{\gamma}{w^{2}-u^{2}}+\frac{\left(u-\mathring{v}\right)^{2}}{u^{2}}$
with its instability branches for $\mathring{v}=1$, $w=1.1$, $\gamma=3$
(horizontal axis \textendash{} $\Re\left\{ u\right\} $, vertical
axis \textendash{} $D$). The instability branches are represented
by $\mathscr{D}\left(\Re\left\{ u\right\} ,\gamma\right)$, where
$u$ is complex characteristic velocity, $\Im\left\{ u\right\} \protect\neq0$,
satisfying equations (\ref{eq:T1B1beta2cj}), that is, $\mathscr{D}\left(u,\gamma\right)=\frac{1}{\breve{\omega}^{2}}$.
Dash-dot (black) lines represent the plot of $\mathscr{D}\left(u,\gamma\right)$
for real $u$, solid (brown) lines represent unstable branches with
$\Im\left\{ u\right\} \protect\neq0$. The characteristic function
instability nodes are represented by solid (brown) diamond dots. Vertical
dash-dot (black) straight lines represent asymptotes associated with
real-valued poles of function $\gamma\mathscr{D}_{\mathrm{T}}\left(u\right)=\frac{\gamma}{w^{2}-u^{2}}$
and consequently of the characteristic function $\mathscr{D}\left(u,\gamma\right)$.}
\end{figure}
\begin{figure}
\begin{centering}
\includegraphics[scale=0.55]{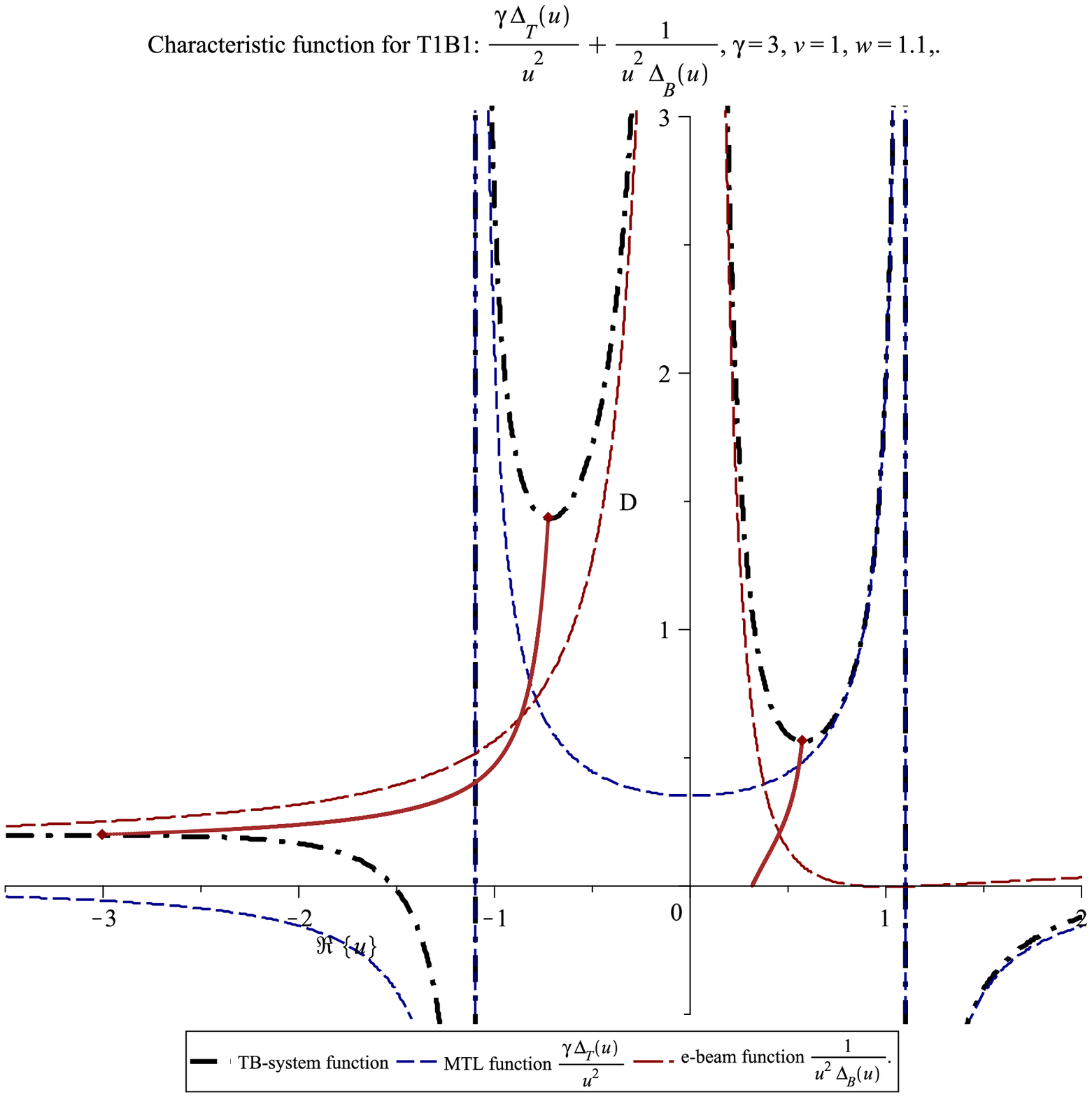}
\par\end{centering}
\centering{}\caption{\label{fig:char-instab-T1B1-gam3-1} Single-stream e-beam coupled
to a single TL. Plot of the characteristic function $\mathscr{D}\left(u,\gamma\right)=\frac{\gamma}{w^{2}-u^{2}}+\frac{\left(u-\mathring{v}\right)^{2}}{u^{2}}$
with its instability branches for $\mathring{v}=1$, $w=1.1$, $\gamma=3$
(horizontal axis \textendash{} $\Re\left\{ u\right\} $, vertical
axis \textendash{} $D$). The instability branches are represented
by $\mathscr{D}\left(\Re\left\{ u\right\} ,\gamma\right)$, where
$u$ is complex characteristic velocity, $\Im\left\{ u\right\} \protect\neq0$,
satisfying equations (\ref{eq:T1B1beta2cj}), that is, $\mathscr{D}\left(u,\gamma\right)=\frac{1}{\breve{\omega}^{2}}$.
Dash-dot (black) lines represent the plot of $\mathscr{D}\left(u,\gamma\right)$
for real $u$, solid (brown) lines represent unstable branches with
$\Im\left\{ u\right\} \protect\neq0$, and dashed (brown, blue) lines
represent, respectively, plots of functions $\gamma\mathscr{D}_{\mathrm{T}}\left(u\right)=\frac{\gamma}{w^{2}-u^{2}}$
and $\mathscr{D}_{\mathrm{B}}\left(u\right)=\frac{\left(u-\mathring{v}\right)^{2}}{u^{2}}$.
The characteristic function instability nodes are represented by solid
(brown) diamond dots. Vertical dash-dot (black) straight lines represent
asymptotes associated with real-valued poles of function $\frac{\gamma}{w^{2}-u^{2}}$
and consequently of the characteristic function $\mathscr{D}\left(u,\gamma\right)$.}
\end{figure}
\begin{figure}[h]
\begin{centering}
\includegraphics[scale=0.3]{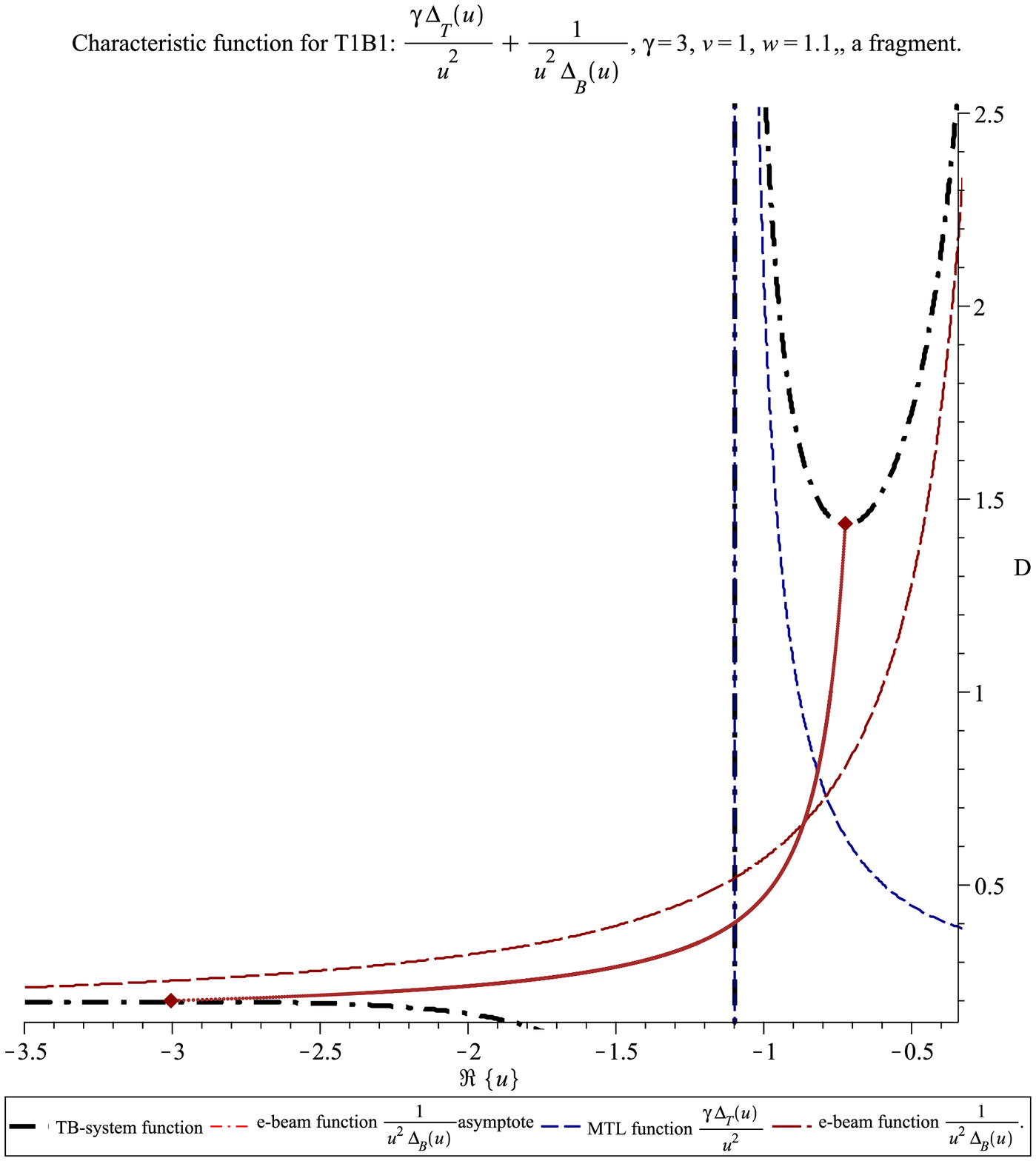}\hspace{0.5cm}\includegraphics[scale=0.32]{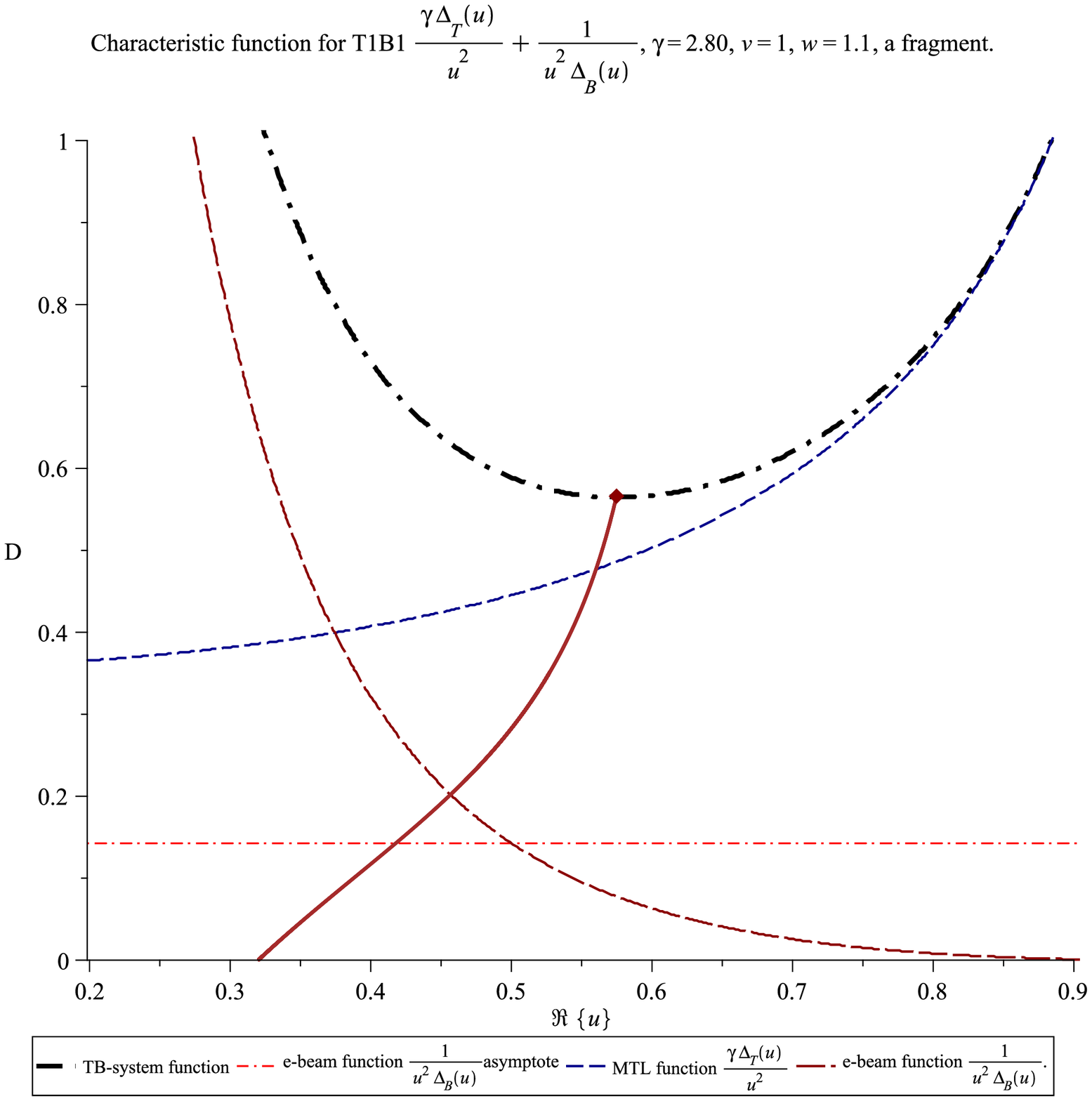}
\par\end{centering}
\centering{}(a)\hspace{5cm}(b)\caption{\label{fig:char-instab-T1B1-gam3-3} Single-stream e-beam coupled
to a single TL. Zoomed fragments of the plot in Fig. \ref{fig:char-instab-T1B1-gam3-1}
of the characteristic function $\mathscr{D}\left(\Re\left\{ u\right\} ,\gamma\right)$
with its instability branches (horizontal axis \textendash{} $\Re\left\{ u\right\} $,
vertical axis \textendash{} $D$) for $\mathring{v}=1$, $w=1.1$,
$\gamma=3$: (a) negative $u$ and $\Re\left\{ u\right\} $; (b) positive
$u$ and $\Re\left\{ u\right\} $. The instability branches are represented
by $\mathscr{D}\left(\Re\left\{ u\right\} ,\gamma\right)$, where
$u$ is complex characteristic velocity, $\Im\left\{ u\right\} \protect\neq0$,
satisfying equations (\ref{eq:T1B1beta2cj}), that is, $\mathscr{D}\left(u,\gamma\right)=\frac{1}{\breve{\omega}^{2}}$.
Dash-dot (black) lines represent the plot of $\mathscr{D}\left(u,\gamma\right)$
for real $u$, solid (brown) lines represent unstable branches with
$\Im\left\{ u\right\} \protect\neq0$, and dashed (brown, blue) lines
represent plots of functions $\gamma\mathscr{D}_{\mathrm{T}}\left(u\right)=\frac{\gamma}{w^{2}-u^{2}}$
and $\mathscr{D}_{\mathrm{B}}\left(u\right)=\frac{\left(u-\mathring{v}\right)^{2}}{u^{2}}$,
respectively. The characteristic function instability nodes are represented
by solid (brown) diamond dots. Vertical dash-dot (black) straight
lines represent the asymptotes associated with real-valued poles of
function $\frac{\gamma}{w^{2}-u^{2}}$ and consequently of the characteristic
function $\mathscr{D}\left(u,\gamma\right)$.}
\end{figure}

\subsection{Dispersion relations\label{subsec:dis-omu}}

Following to \cite[4]{FigTWTbk} we define for any characteristic
velocity $u$ from characteristic velocity set $\mathscr{U}_{\mathrm{TB}}^{+}\left(\bar{\gamma}\right)$
as in equations (\ref{eq:UTBplus})\emph{ the TWT frequency function}
$\Omega\left(u\right)$ so that it is the frequency corresponding
to $u$ in the characteristic equation (\ref{eq:DsT1B1cj}), namely
\begin{equation}
\Omega\left(u\right)=\frac{\Omega_{\mathrm{p}}}{\sqrt{\mathscr{D}\left(u,\bar{\gamma}\right)}}>0,\quad\Omega_{\mathrm{p}}=R_{\mathrm{sc}}\omega_{\mathrm{p}},\quad u\in\mathscr{U}_{\mathrm{TB}}^{+}\left(\bar{\gamma}\right),\quad\mathscr{D}\left(u,\bar{\gamma}\right)>0.\label{eq:omOmu1aj}
\end{equation}
Notice that by the definition of the characteristic velocity $u$,
the value of $\mathscr{D}\left(u,\bar{\gamma}\right)$ must be real
and positive, and the square root $\sqrt{\mathscr{D}\left(u,\bar{\gamma}\right)}$
is assumed to be positive as well. Consequently, $\Omega\left(u\right)$
is real-valued. Also note that if $u_{0}$ is a real-valued characteristic
velocity, then for any real $u$ belonging to its sufficiently small
vicinity, we have $\mathscr{D}\left(u,\bar{\gamma}\right)>0$, implying
that it is also a characteristic velocity. Consequently, $\Omega\left(u\right)$
defined by equation (\ref{eq:omOmu1aj}) is a real analytic function
in the vicinity of $u_{0}$. We may also view the equation
\begin{equation}
\Omega\left(u\right)=\frac{\Omega_{\mathrm{p}}}{\sqrt{\mathscr{D}\left(u,\bar{\gamma}\right)}}=\omega\label{eq:omOmu1aomj}
\end{equation}
as an equivalent form of the characteristic equation (\ref{eq:DsT1B1cj})
for $u$ in the vicinity of $u_{0}$. Equation (\ref{eq:omOmu1aomj})
can be viewed also as a form of the dispersion relations and refer
to it \emph{velocity dispersion relation relation}. Using the frequency
function $\Omega\left(u\right)$, we naturally define \emph{the TWT
wavenumber function} $K\left(u\right)$ by the following equation:
\begin{equation}
K\left(u\right)=\frac{\Omega\left(u\right)}{u},\quad u\in\mathscr{U}_{\mathrm{TB}}^{+}\left(\bar{\gamma}\right),\quad\Omega\left(u\right)>0.\label{eq:omOmu1aK}
\end{equation}

To visualize features of the TWT instability in its dispersion relation
we proceed as follows \cite[7]{FigTWTbk}. We represent the set of
all oscillatory and unstable modes of the TWT-system geometrically
by the set $\Pi_{\mathrm{TB}}$ of the corresponding modal points
$\left(k\left(\omega\right),\omega\right)$ and $\left(\Re\left\{ k\left(\omega\right)\right\} ,\omega\right)$
in the $k\omega$-plane. We name the set $\Pi_{\mathrm{TB}}$ as the
\emph{dispersion-instability graph}. To distinguish graphically points
$\left(k\left(\omega\right),\omega\right)$ associated oscillatory
modes when $k\left(\omega\right)$ is real-valued from points $\left(\Re\left\{ k\left(\omega\right)\right\} ,\omega\right)$
associated unstable modes when $k\left(\omega\right)$ is complex-valued
with $\Im\left\{ k\left(\omega\right)\right\} \neq0$, we assign them
colors as follows: (i) blue color is assigned to points $\left(k\left(\omega\right),\omega\right)$
when $k$ is real-valued; (ii) brown color is assigned to points $\left(\Re\left\{ k\left(\omega\right)\right\} ,\omega\right)$
when $k\left(\omega\right)$ is complex-valued with $\Im\left\{ k\left(\omega\right)\right\} \neq0$.
We remind once again that every brown point $\left(\omega,\Re\left\{ k\left(\omega\right)\right\} \right)$
represents exactly two complex conjugate unstable modes associated
with $\pm\Im\left\{ k\left(\omega\right)\right\} $. We introduce
then two subsets of the set $\Pi_{\mathrm{TB}}$, namely, $\Pi_{\mathrm{TBo}}$
and $\Pi_{\mathrm{TBu}}$ representing, respectively, all oscillatory
and unstable modes. Evidently, the sets $\Pi_{\mathrm{TBo}}$ and
$\Pi_{\mathrm{TBu}}$ are disjoint, and they partition the set $\Pi_{\mathrm{TB}}$,
that is, $\Pi_{\mathrm{TB}}=\Pi_{\mathrm{TBo}}\cup\Pi_{\mathrm{TBu}}$.
According to the color assignments, the points of sets $\Pi_{\mathrm{TBo}}$
and $\Pi_{\mathrm{TBu}}$ are, respectively, blue and brown. Fig.
\ref{fig:dis-instab-T1B1-gam3} shows a typical dispersion-instability
graph for the data as in equation (\ref{eq:datavwgam}).
\begin{figure}[h]
\begin{centering}
\includegraphics[scale=0.45]{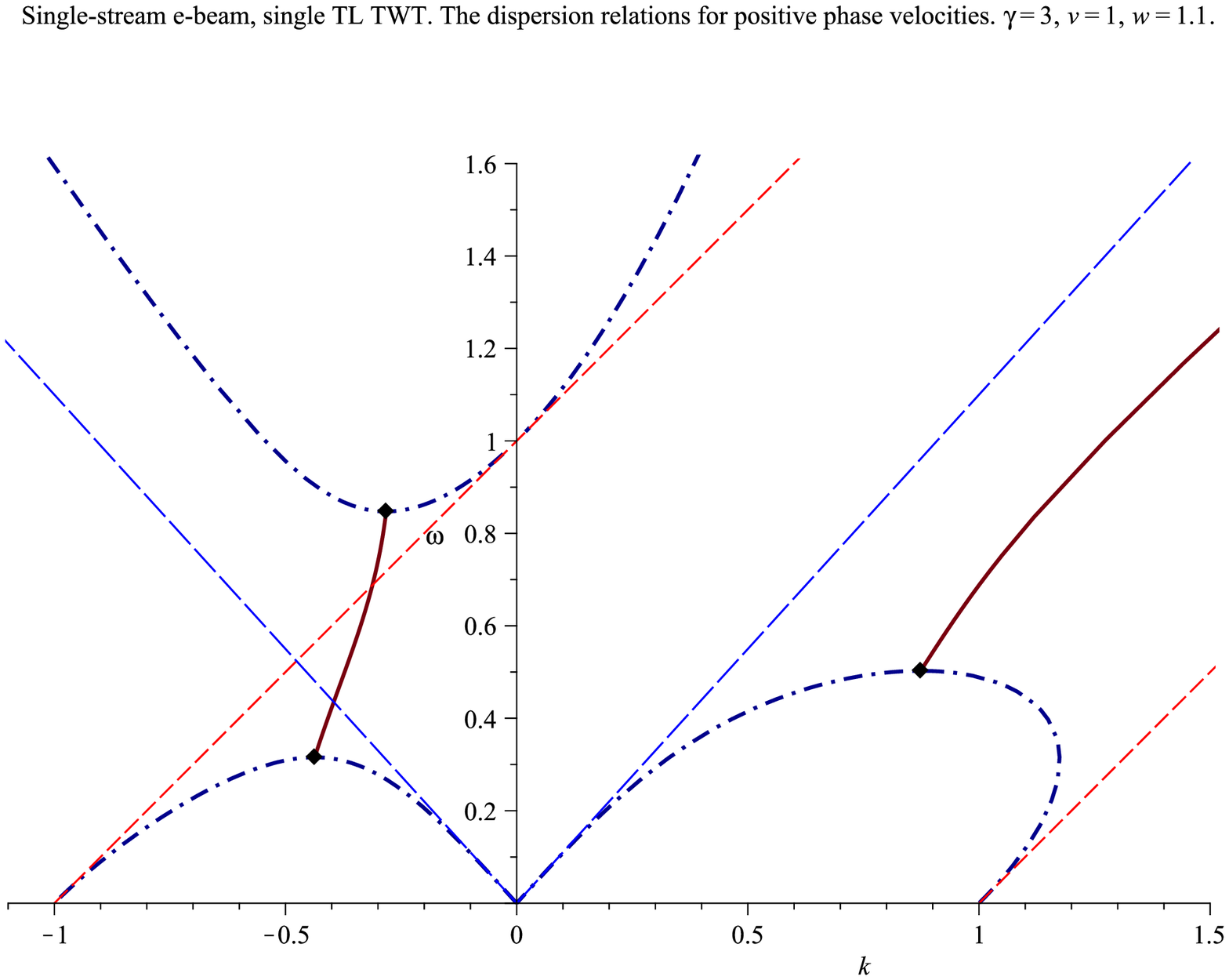}
\par\end{centering}
\centering{}\caption{\label{fig:dis-instab-T1B1-gam3} Single-stream e-beam coupled to
a single TL. Dispersion-instability graph (horizontal axis \textendash $\Re\left\{ k\right\} $,
vertical axis \textendash{} $\omega$) for $\mathring{v}=1$, $w=1.1$,
$\gamma=3$. Solid (brown) curves represent unstable branches with
$\Im\left\{ k\right\} \protect\neq0$, dash-dotted (black) lines represent
oscillatory branches with $\Im\left\{ k\right\} =0$, dashed lines
represent the dispersion relations for uncoupled TL (blue, converging
to the origin) and for an uncoupled e-beam (red). The instability
nodes are represented by solid (brown) diamond dots.}
\end{figure}

\subsection{Nodal velocities, nodal function and equation\label{subsec:nod-func}}

We consider here the \emph{nodal velocities} which are the phase velocities
that signify the onset of the TWT instability. According to our theory
developed in \cite[13, 30]{FigTWTbk} the nodal velocities are solutions
to the nodal equation
\begin{equation}
\mathscr{R}\left(u\right)=-\frac{\partial_{u}\mathscr{D}_{\mathrm{B}}\left(u\right)}{\partial_{u}\mathscr{D}_{\mathrm{T}}\left(u\right)}=\frac{\mathring{v}\left(\mathring{v}-u\right)\left(w^{2}-u^{2}\right)^{2}}{u^{4}}=\gamma,\label{eq:RuT1B1aj}
\end{equation}
where functions $\mathscr{D}_{\mathrm{T}}\left(u\right)$ and $\mathscr{D}_{\mathrm{B}}\left(u\right)$
are defined by equations (\ref{eq:DsT1B1bj}) and $\gamma$ is TWT
principal parameter defined equations (\ref{eq:DsT1B1hj}). We refer
to function $\mathscr{R}\left(u\right)$ in the \emph{nodal function}.
Figure \ref{fig:T1B1-node1} illustrates graphically features of the
nodal function $\mathscr{R}\left(u\right)$ for a TWT composed of
a single-stream e-beam and a single TL. Equation (\ref{eq:RuT1B1aj})
can be recast in terms of dimensionless variables as follows:
\begin{equation}
\mathscr{R}\left(\check{u}\right)=\frac{\left(1-\check{u}\right)\left(\chi^{2}-\check{u}^{2}\right)^{2}}{\check{u}^{4}}=\check{\gamma},\quad\check{\gamma}=\frac{\gamma}{\mathring{v}{}^{2}},\quad\check{u}=\frac{u}{\mathring{v}},\quad\chi=\frac{w}{\mathring{v}},\label{eq:RT1B1nod1aj}
\end{equation}
where $\check{\gamma}$ is the dimensionless TWT principal parameter.
If $\mathsf{u}_{0}$ is the nodal velocity satisfying equation (\ref{eq:RT1B1nod1aj}),
then the corresponding instability node frequency $\omega_{0}$ as
well as instability node wavenumber $k_{0}$ can be expressed in terms
of the frequency function $\Omega\left(u\right)$ defined by equation
(\ref{eq:omOmu1aj}), that is,
\begin{equation}
\omega_{0}=\Omega\left(\mathsf{u}_{0}\right)=\frac{\Omega_{\mathrm{p}}}{\sqrt{\mathscr{D}\left(\mathsf{u}_{0},\bar{\gamma}\right)}},\quad k_{0}=\frac{\omega_{0}}{\mathsf{u}_{0}},\label{eq:RT1B1nod1bj}
\end{equation}
where $\Omega\left(u\right)$ is defined by equations (\ref{eq:omOmu1aj}).

Observe that the nodal equation (\ref{eq:RuT1B1aj}) has exactly three
real-valued solutions: (i) one positive $\mathsf{u}^{+}\left(\bar{\gamma}\right)>0$;
(ii) two negative $\mathsf{u}_{1}^{-}\left(\bar{\gamma}\right)<\mathsf{u}_{2}^{-}\left(\bar{\gamma}\right)<0$.
Notice also since in equation (\ref{eq:omOmu1aj}) $\gamma>0$ then
$\mathsf{u}^{+}\left(\bar{\gamma}\right)<1$ and hence we have
\begin{equation}
0<\mathsf{u}^{+}\left(\bar{\gamma}\right)<1.\label{eq:upgam1a}
\end{equation}
 
\begin{figure}[h]
\begin{centering}
\hspace{-0.5cm}\includegraphics[scale=0.3]{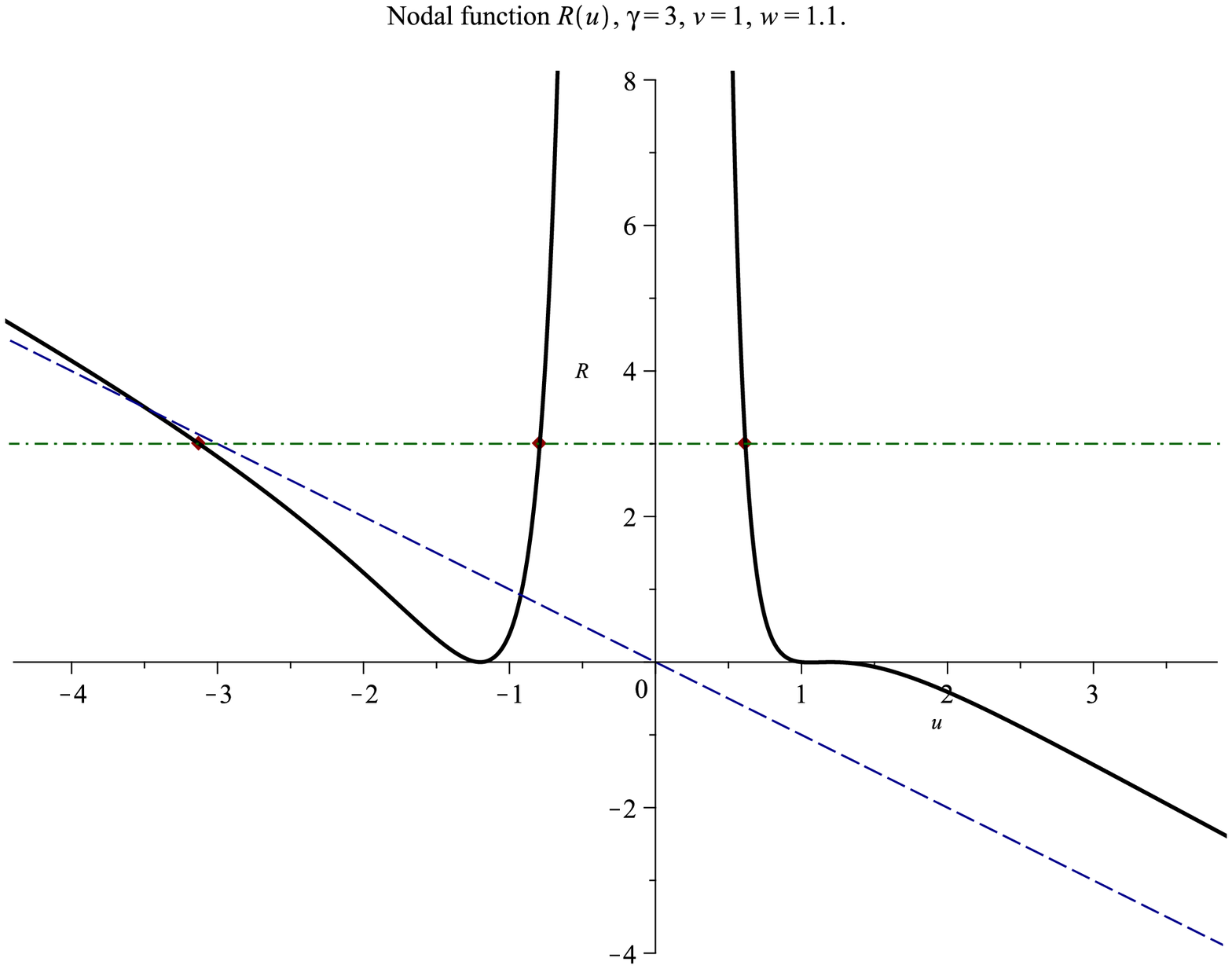}\hspace{0.8cm}\includegraphics[scale=0.24]{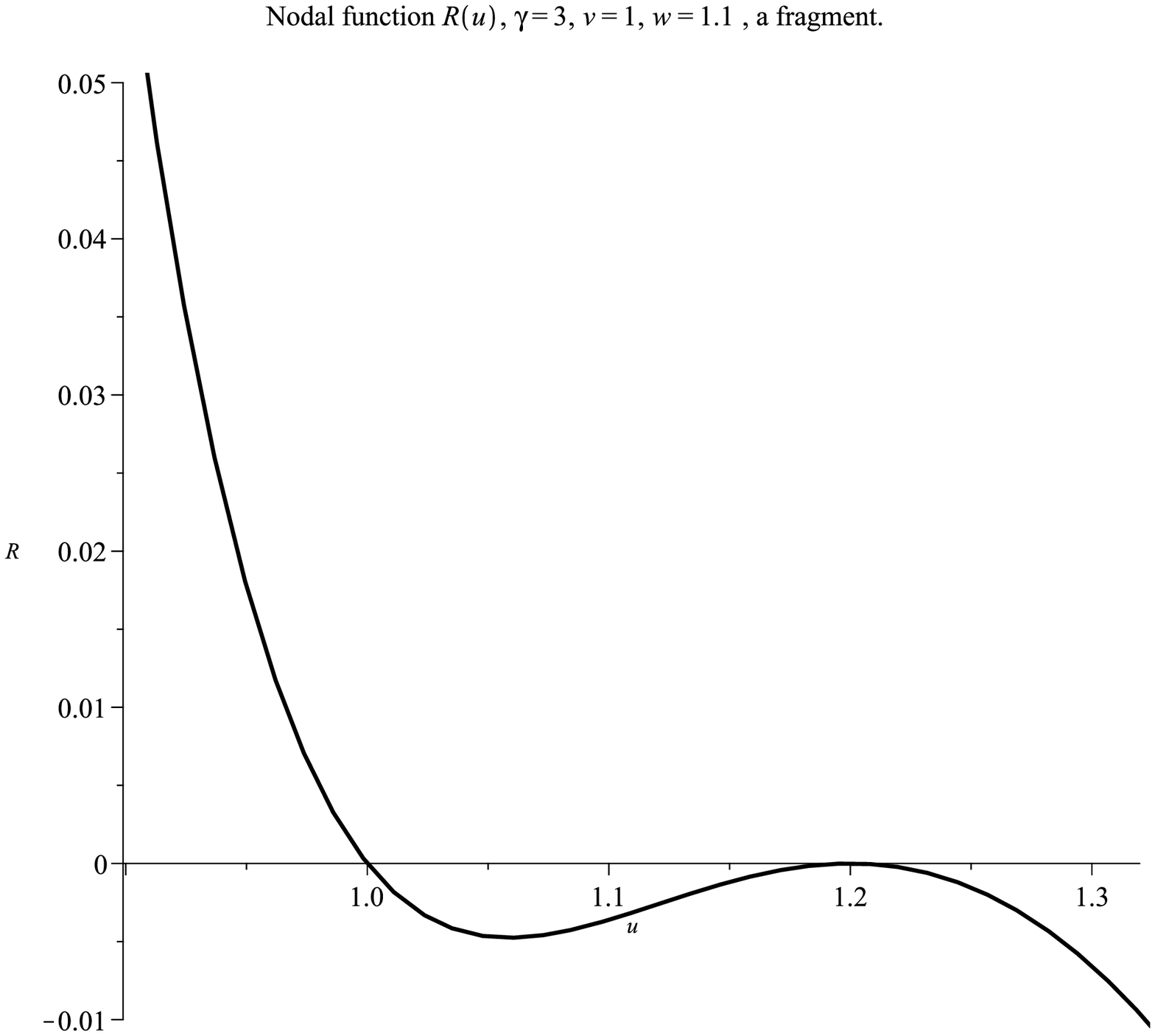}
\par\end{centering}
\centering{}(a)\hspace{5cm}(b)\caption{\label{fig:T1B1-node1} Single e-beam coupled to a single TL. Plot
(a) of the nodal function $\mathscr{R}\left(u\right)=\frac{v\left(v-u\right)\left(w^{2}-u^{2}\right)^{2}}{u^{4}}$
and (b) its fragment (horizontal axis \textendash{} $u$, vertical
axis \textendash{} $R$) for $\mathring{v}=1$, $w=1.1$. Solid (black)
line represents function $\mathscr{R}\left(u\right)$, and dashed
(blue) line represents asymptotics function $\mathscr{R}\left(u\right)$
as $u\rightarrow\infty$, the dash-dot (green) line represents constant
$\gamma=3$, and solid (brown) diamond dots represent solutions for
the instability node equation (\ref{eq:RuT1B1aj}), that is, $\mathscr{R}\left(u\right)=\gamma=3.$}
\end{figure}

\subsection{Nodal points as points of degeneracy of the dispersion relations\label{subsec:nod-deg}}

The diamond (brown) dots Fig. \ref{fig:dis-instab-T1B1-gam3} identify
the three points of the degeneracy of the dispersion relations. Vicinites
of these points demonstrate graphically the transition to instability
when points on dot-dashed (blue) lines associated with oscillatory
modes with real phase velocities $u$ lines and real wavenumbers $k$
merge with points on solid (brown) lines associated with complex-conjugate
pairs of unstable modes with non-real real phase velocities at nontrivial
$k\omega$-nodes.

Importantly according to our studies in \cite[13, 30]{FigTWTbk} the
nodal velocities are in fact the points of degeneracy of the dispersion
relation. In addition to that,\emph{ any nodal velocity $\mathsf{u}_{0}$
has to be a real number which is an extreme point of characteristic
function $\mathscr{D}\left(u,\bar{\gamma}\right)$}. Consequently,
$\mathsf{u}_{0}$ has to be an extreme point function $\Omega\left(u\right)$
defined by equation (\ref{eq:omOmu1aj}). Based on this, we can assume
that
\begin{equation}
\left.\partial_{u}\Omega\left(u\right)\right|_{\mathsf{u}_{0}}=0,\quad\left.\partial_{u}^{2}\Omega\left(u\right)\right|_{\mathsf{u}_{0}}\neq0,\label{eq:omOmu1bj}
\end{equation}
that is, $\mathsf{u}_{0}$ is a point of degeneracy of second-order
of function $\Omega\left(u\right)$.

To assess the analytic properties of solutions for the characteristic
equation (\ref{eq:DsT1B1cj}) in the vicinity of $\mathsf{u}_{0}$,
we introduce the following ``small'' dimensionless parameters, which
are useful for our studies of local extrema of function $\Omega\left(u\right)$
defined by equations (\ref{eq:omOmu1aj}):
\begin{gather}
\delta_{u}=\frac{u-\mathsf{u}_{0}}{\mathsf{u}_{0}},\quad\delta_{\omega}=\frac{\omega-\omega_{0}}{\omega_{0}},\quad\delta_{k}=\frac{k-k_{0}}{k_{0}},\quad\kappa=\Re\left\{ \delta_{k}\right\} ,\label{eq:omOmu2aj}\\
\eta=\left[\xi_{0}^{-1}\delta_{\omega}\right]^{\frac{1}{2}}=\left[\frac{\omega-\omega_{0}}{\omega_{2}}\right]^{\frac{1}{2}},\quad\xi_{0}=\frac{\omega_{2}}{\omega_{0}},\quad\omega_{2}=\frac{\mathsf{u}_{0}^{2}}{2}\left.\partial_{u}^{2}\Omega\left(u\right)\right|_{\mathsf{u}_{0}}\neq0,\label{eq:omOmu2bj}
\end{gather}
where the square roots in equation (\ref{eq:omOmu2bj}) are naturally
defined up to a factor $\pm1$. Notice that parameter $\xi_{0}$ can
positive or negative depending on the sign of $\omega_{2}$. Observe
also that $\delta_{k}$ can be expressed in terms of $\eta$ and $\delta_{u}$
by the following formulas:
\begin{gather}
\delta_{k}=\frac{k-k_{0}}{k_{0}}=\frac{1+\delta_{\omega}}{1+\delta_{u}}-1=\frac{1+\xi_{0}\eta^{2}}{1+\delta_{u}}-1,\label{eq:omOmu2cj}\\
k=\frac{\omega}{u},\quad k_{0}=\frac{\omega_{0}}{\mathsf{u}_{0}}.\label{eq:omOmu2dj}
\end{gather}
Then using the analyticity of function $\Omega\left(u\right)$, we
introduce its power series
\begin{gather}
\Omega\left(u\right)=\omega_{0}+\omega_{2}\delta_{u}^{2}\left[1+\sum_{n\geq3}\Omega_{n}\delta_{u}^{n-2}\right],\quad\delta_{u}=\frac{u-\mathsf{u}_{0}}{\mathsf{u}_{0}},\label{eq:omOmu1cj}\\
\omega_{2}=\frac{\mathsf{u}_{0}^{2}}{2}\left.\partial_{u}^{2}\Omega\left(u\right)\right|_{\mathsf{u}_{0}}\neq0,\quad\Omega_{n}=\frac{\mathsf{u}_{0}^{n}}{\omega_{2}}\frac{\left.\partial_{u}^{n}\Omega\left(u\right)\right|_{\mathsf{u}_{0}}}{n!},\quad n\geq3,\nonumber 
\end{gather}
where $\omega_{0}$ and $\omega_{2}$ have the physical dimensions
of frequency, whereas variable $\delta_{u}$ and coefficients $\Omega_{n}$
are dimensionless. Notice that coefficients $\omega_{0}$, $\omega_{2}$
and $\Omega_{n}$ are all real.

Using dimensionless variables (\ref{eq:omOmu2cj}) and series (\ref{eq:omOmu1cj})
we can recast equation (\ref{eq:omOmu1aj}) as follows:
\begin{equation}
\delta_{\omega}=\xi_{0}\delta_{u}^{2}\left[1+\sum_{n\geq3}\Omega_{n}\delta_{u}^{n-2}\right].\label{eq:omOmu1aaj}
\end{equation}
The analysis of equations (\ref{eq:omOmu2aj})-(\ref{eq:omOmu1aaj})
carried out in \cite[13, 54]{FigTWTbk} provides the following power
series approximation that relates $\delta_{\omega}$ and $\kappa=\Re\left\{ \delta_{k}\right\} $:
\begin{gather}
\delta_{\omega}=\frac{\omega-\omega_{0}}{\omega_{0}}=\left\{ \begin{array}{rcl}
\xi_{0}\kappa^{2}+\cdots, & \text{if} & \xi_{0}^{-1}\delta_{\omega}>0\\
\frac{2\xi_{0}}{2+2\xi_{0}+\Omega_{3}}\kappa+\cdots, & \text{ if } & \xi_{0}^{-1}\delta_{\omega}<0
\end{array}\right.,\quad\kappa=\Re\left\{ \delta_{k}\right\} =\Re\left\{ \frac{k-k_{0}}{k_{0}}\right\} .\label{eq:omOmu6j}
\end{gather}
Solutions to characteristic equation (\ref{eq:DsT1B1cj}) or the equivalent
to it equation (\ref{eq:omOmu1aomj}) can be transformed into dispersion
relation $\omega=\omega\left(\Re\left\{ k\right\} \right)$, which
is represented graphically in Fig. \ref{fig:instab-approxj} for the
data as in equation (\ref{eq:datavwgam}).
\begin{figure}
\begin{centering}
\includegraphics[scale=0.4]{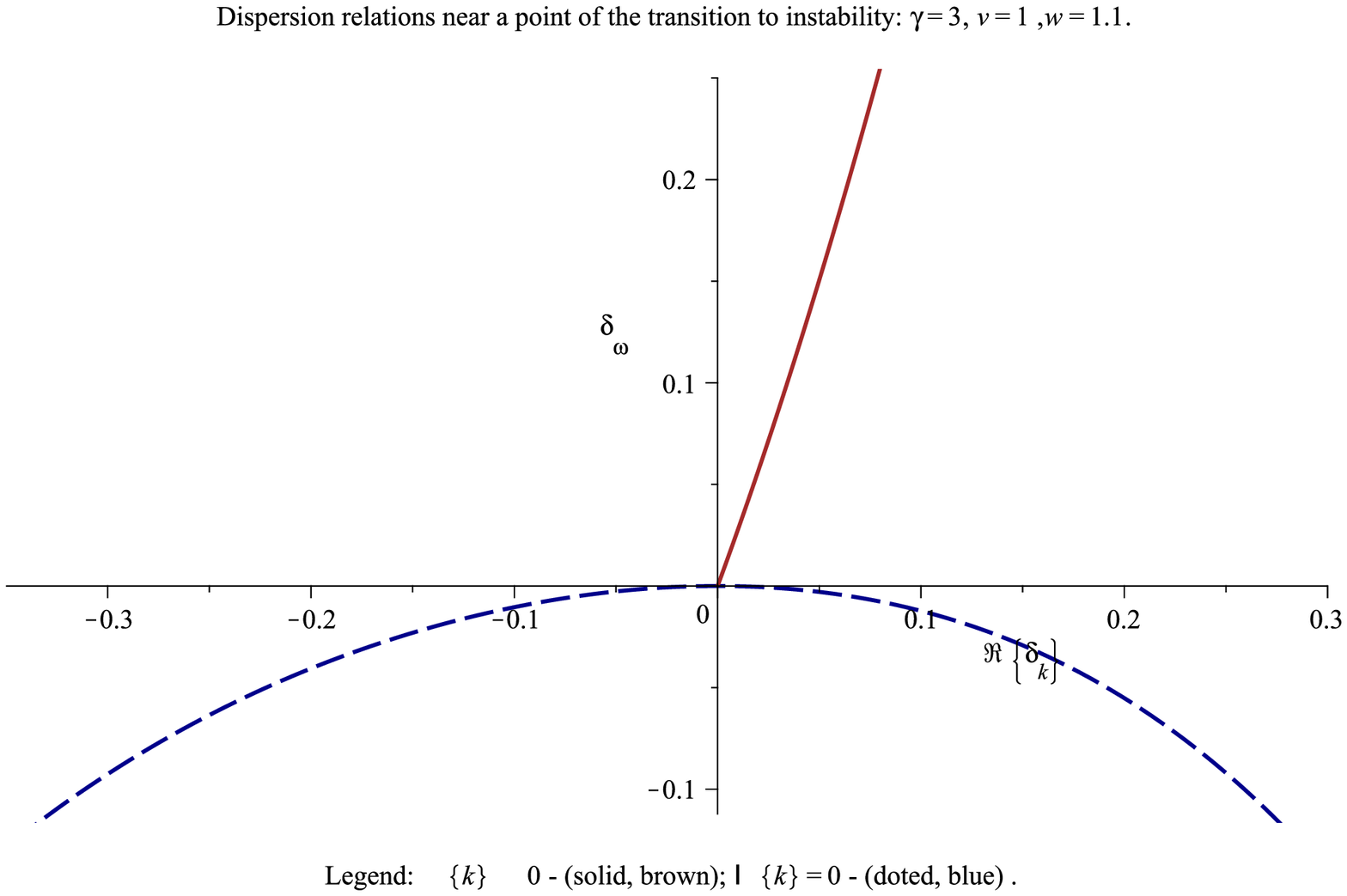}
\par\end{centering}
\centering{}\caption{\label{fig:instab-approxj} A fragment of the dispersion-instability
graph in a vicinity of nodal velocity $\mathsf{u}^{+}\left(\bar{\gamma}\right)$
computed based on asymptotic formulas (\ref{eq:omOmu6j}) for $\mathring{v}=1$,
$w=1.1$, $\gamma=3$ (horizontal axis \textendash{} $\Re\left\{ \delta_{k}\right\} $,
vertical axis \textendash{} $\delta_{\omega}$). Solid (brown) lines
represent unstable branches with $\Im\left\{ k\right\} \protect\neq0$,
and dotted (blue) lines represent oscillatory branches with $\Im\left\{ k\right\} =0$.}
\end{figure}

\section{Analysis of the TWT at the EPD\label{sec:twt-epd}}

The analysis of the TWT at the EPD we carry out here is based on the
TWT properties reviewed in Section \ref{sec:twt-mod}. In particular,
we take a close look into the spectral properties of the TWT at the
EPD associated with the nodal velocity $\mathsf{u}^{+}\left(\bar{\gamma}\right)$
considered in Section \ref{subsec:nod-func}. These properties are
related to the TWT monic matrix polynomial $\mathsf{M}_{u\omega}$
defined in equations (\ref{eq:MLkC2c}), that is
\begin{gather}
\mathsf{M}_{u\omega}=\left[\begin{array}{rr}
\check{u}{}^{2}-\chi^{2} & -\chi^{2}\\
\frac{\check{\gamma}}{\check{\omega}{}^{-2}-1} & \check{u}{}^{2}+\frac{2\check{u}}{\check{\omega}{}^{-2}-1}+\frac{\check{\gamma}-1}{\check{\omega}{}^{-2}-1}
\end{array}\right].\label{eq:CMuom1a}
\end{gather}
for $\check{u}=\mathsf{u}^{+}\left(\bar{\gamma}\right)$, and to the
corresponding companion matrix we consider in Section \ref{subsec:comp-mat}.
For the reader's convenience the basics of the theory of matrix polynomials
and the corresponding companion matrices are reviewed in Section \ref{sec:mat-poly}.

In turns out that it is advantageous for our purposes here to use
the nodal velocity
\begin{equation}
p=\mathsf{u}^{+}\left(\bar{\gamma}\right)>0\label{eq:pugam1a}
\end{equation}
as a \emph{new independent variable in place of parameter $\bar{\gamma}$}.
Using the fact that $\mathsf{u}^{+}\left(\bar{\gamma}\right)$ is
a monotonically decreasing function of $\bar{\gamma}$ we can uniquely
recover $\bar{\gamma}$ from $p$ based on equation (\ref{eq:pugam1a}).
We naturally refer to $p$ and the\emph{ nodal velocity} for: (i)
it is a point of degeneracy for the dispersion relations associated
with the nodal velocity $\mathsf{u}^{+}\left(\bar{\gamma}\right)$
and (ii) it is also a point of local minimum of the characteristic
function $\mathscr{D}\left(\check{u},\check{\gamma}\right)$ defined
by equation (\ref{eq:DsT1B1dj}), see Figs. \ref{fig:char-instab-T1B1-gam3-3}
(b), \ref{fig:dis-instab-T1B1-gam3} and \ref{fig:T1B1-node1} (a)
for graphical illustration.

\subsection{Nodal velocity as the TWT parameter\label{subsec:nod-vel-par}}

Plugging $\check{u}=p$ in the nodal equation (\ref{eq:RT1B1nod1aj})
we obtain the following representation of $\check{\gamma}$ in terms
$p$:
\begin{equation}
\check{\gamma}=\check{\gamma}_{\mathrm{e}}=\check{\gamma}_{\mathrm{e}}\left(p,\chi\right)=\frac{\left(1-p\right)\left(p^{2}-{\it \chi}^{2}\right)^{2}}{p^{4}}>0,\quad0<p<1.\label{eq:gamdp1a}
\end{equation}
Notice that equation (\ref{eq:gamdp1a}) can be viewed as the inversion
of equation (\ref{eq:pugam1a}) that defines the nodal velocity $p=\mathsf{u}^{+}\left(\bar{\gamma}\right)$.
Under assumption that $\chi$ is given and fixed function $\check{\gamma}_{\mathrm{e}}\left(p,\chi\right)$
in equation (\ref{eq:gamdp1a}) relates the nodal velocity $p$ to
the TWT principle parameter $\check{\gamma}$. Notice that $\check{\gamma}_{\mathrm{e}}\left(p,\chi\right)$
is a monotonically decreasing function of $p$ as illustrated by Figure
\ref{fig:gam-p-T1B1}. The function monotonicity can be established
by an examination of the sign of its derivative, which is
\begin{equation}
\partial_{p}\left(\check{\gamma}_{\mathrm{e}}\left(p,\chi\right)\right)=\frac{\left({\it \chi}^{2}-p^{2}\right)\left(3{\it \chi}^{2}p+p^{3}-4{\it \chi}^{2}\right)}{p^{5}}<0,\text{for }0<p<1<\chi.\label{eq:gamdp1b}
\end{equation}
Since factor $3{\it \chi}^{2}p+p^{3}-4{\it \chi}^{2}$ in the right-hand
side of equation (\ref{eq:gamdp1b}) is a monotonically growing function
of $p$ it satisfies the following inequalities
\begin{equation}
-4{\it \chi}^{2}<3{\it \chi}^{2}p+p^{3}-4{\it \chi}^{2}\leq1-{\it \chi}^{2}<0,\text{for }0<p<1<\chi,\label{eq:gamdp1c}
\end{equation}
implying that $\check{\gamma}_{\mathrm{e}}\left(p,\chi\right)$ is
indeed a monotonically decreasing function of $p$.

Typical values of the TWT principal parameter $\check{\gamma}$ are
small and according to \cite[Remark 62.1]{FigTWTbk} they vary between
$2\cdot10^{-6}$ and $0.00675$. Small values of $\check{\gamma}$
according to relation (\ref{eq:gamdp1a}) correspond to values of
the nodal velocity $p$ that are close to 1, and the following asymptotic
formula holds
\begin{equation}
\check{\gamma}_{\mathrm{e}}\left(p,\chi\right)=\left({\it \chi}^{2}-1\right)^{2}\left(1-p\right)+O\left(\left(1-p\right)\right),\quad p\rightarrow1-0.\label{eq:gamdp1d}
\end{equation}
Figure \ref{fig:gam-p-T1B1} shows the graph of function $\check{\gamma}_{\mathrm{e}}\left(p,\chi\right)$
as in equation (\ref{eq:gamdp1a}) for small values of $\check{\gamma}$.
\begin{figure}[h]
\begin{centering}
\hspace{-0.5cm}\includegraphics[scale=0.3]{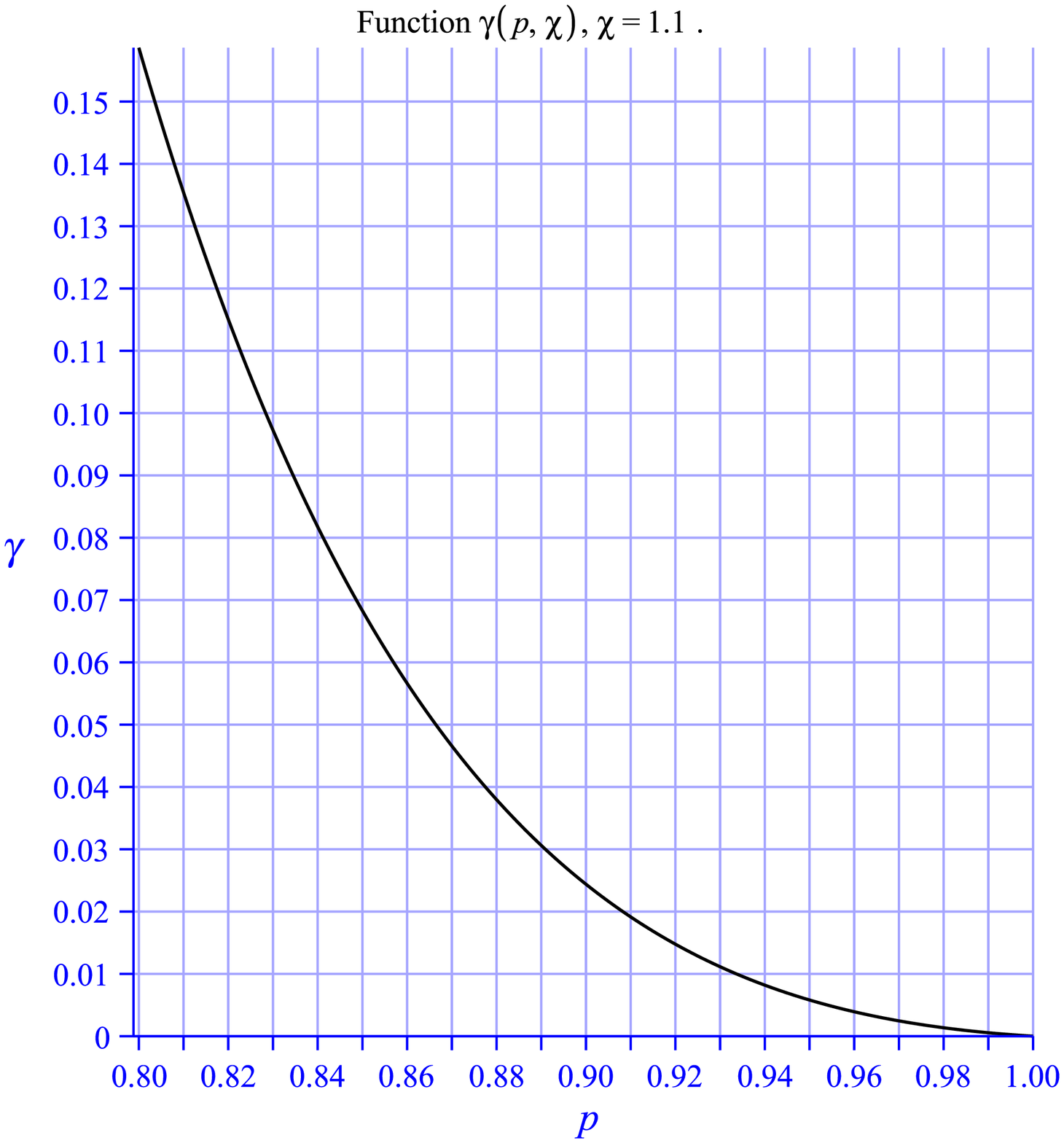}\hspace{0.8cm}\includegraphics[scale=0.3]{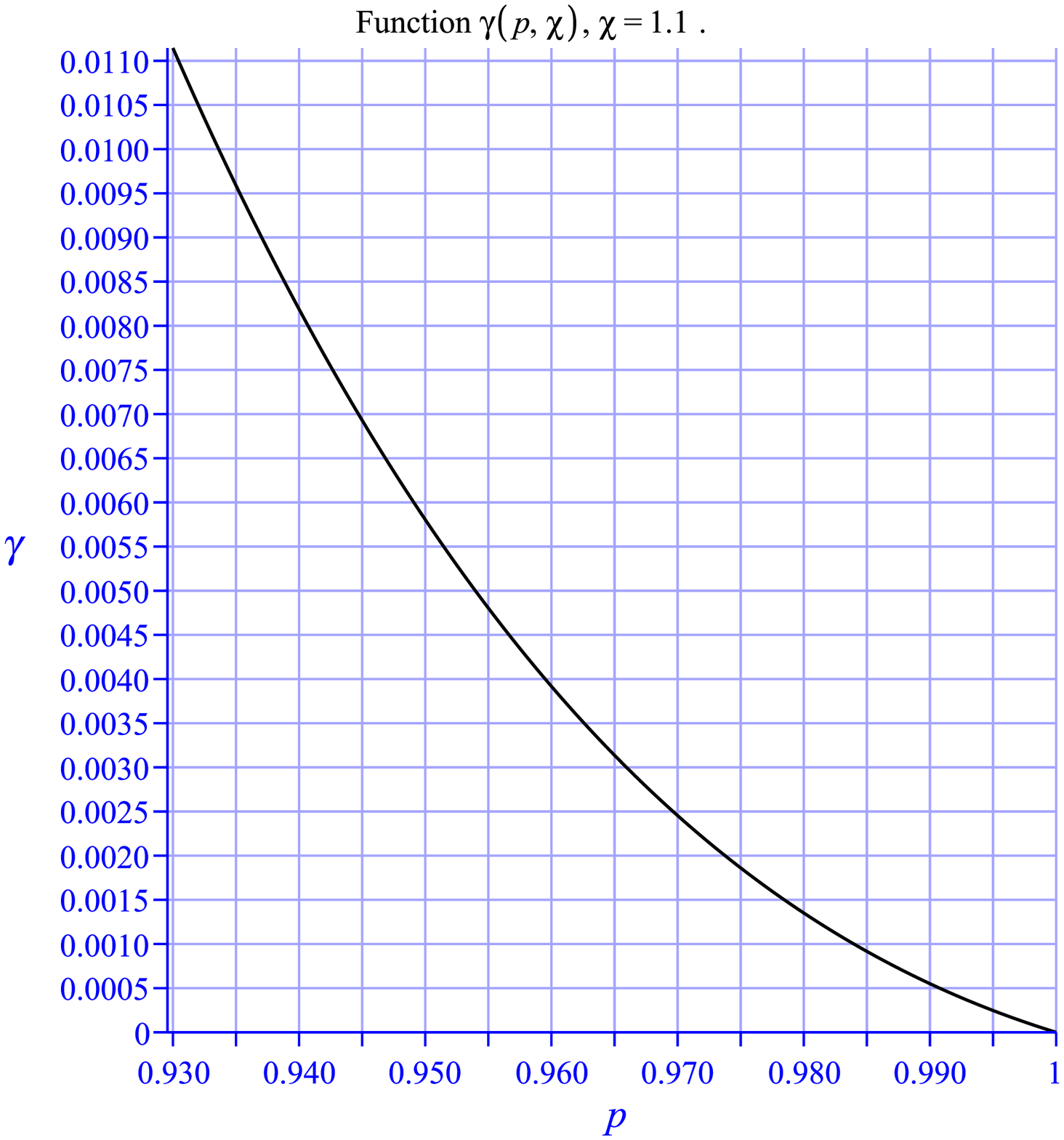}
\par\end{centering}
\centering{}(a)\hspace{6.5cm}(b)\caption{\label{fig:gam-p-T1B1} Plots of function $\check{\gamma}_{\mathrm{e}}\left(p,\chi\right)$
as in (\ref{eq:gamdp1a}) for $\chi=1.1$ and (a) $0.8<p<1$; (b)
$0.93<p<1$.}
\end{figure}

Using equations (\ref{eq:DsT1B1dj}) and (\ref{eq:gamdp1a}) we obtain
the following representation of the characteristic function $\mathscr{D}$
and the corresponding characteristic equation:
\begin{equation}
\mathscr{D}_{\mathrm{e}}\left(\check{u},p\right)=\frac{\left(p-1\right)\left(p^{2}-{\it \chi}^{2}\right)^{2}}{\left(\check{u}^{2}-{\it \chi}^{2}\right)p^{4}}+\frac{\left(\check{u}-1\right)^{2}}{\check{u}{}^{2}}=\frac{1}{\check{\omega}^{2}}.\label{eq:Dsup1a}
\end{equation}
In particular, plugging in $\check{u}=p$ in $\mathscr{D}\left(\check{u},p\right)$
we obtain
\begin{equation}
\mathscr{D}_{\mathrm{e}}\left(p,p\right)=\frac{\left(1-p\right)\left({\it \chi}^{2}-p^{3}\right)}{p^{4}}>0,\text{for }0<p<1<\chi.\label{eq:Dsup1b}
\end{equation}
Hence we can introduce now a positive frequency $\check{\omega}_{\mathrm{e}}=\check{\omega}_{\mathrm{e}}\left(p,\chi\right)$
by the following equality
\begin{equation}
\frac{1}{\check{\omega}_{\mathrm{e}}^{2}}=\mathscr{D}_{\mathrm{e}}\left(p,p\right)=\frac{\left(1-p\right)\left({\it \chi}^{2}-p^{3}\right)}{p^{4}},\label{eq:Dsup1c}
\end{equation}
or equivalently
\begin{equation}
\check{\omega}_{\mathrm{e}}=\check{\omega}_{\mathrm{e}}\left(p,\chi\right)=\frac{1}{\sqrt{\mathscr{D}_{\mathrm{e}}\left(p,p\right)}}=\frac{p^{2}}{\sqrt{\left(1-p\right)\left({\it \chi}^{2}-p^{3}\right)}}>0,\text{for }0<p<1<\chi.\label{eq:Dsup1d}
\end{equation}
We refer to $\check{\omega}_{\mathrm{e}}=\check{\omega}_{\mathrm{e}}\left(p,\chi\right)$
as\emph{ EPD frequency}. In view of equations (\ref{eq:Dsup1a}) and
(\ref{eq:gamdp1b}) the EPD frequency $\check{\omega}_{\mathrm{e}}\left(p,\chi\right)$
defined by equations (\ref{eq:Dsup1d}) corresponds to the nodal velocity
$p$. Figure \ref{fig:om-p-T1B1} is a graphical representation of
function $\check{\omega}_{\mathrm{e}}\left(p,\chi\right)$ as in equation
(\ref{eq:Dsup1d}) for different ranges of values of the nodal velocity
$p$.
\begin{figure}[h]
\begin{centering}
\hspace{-0.5cm}\includegraphics[scale=0.35]{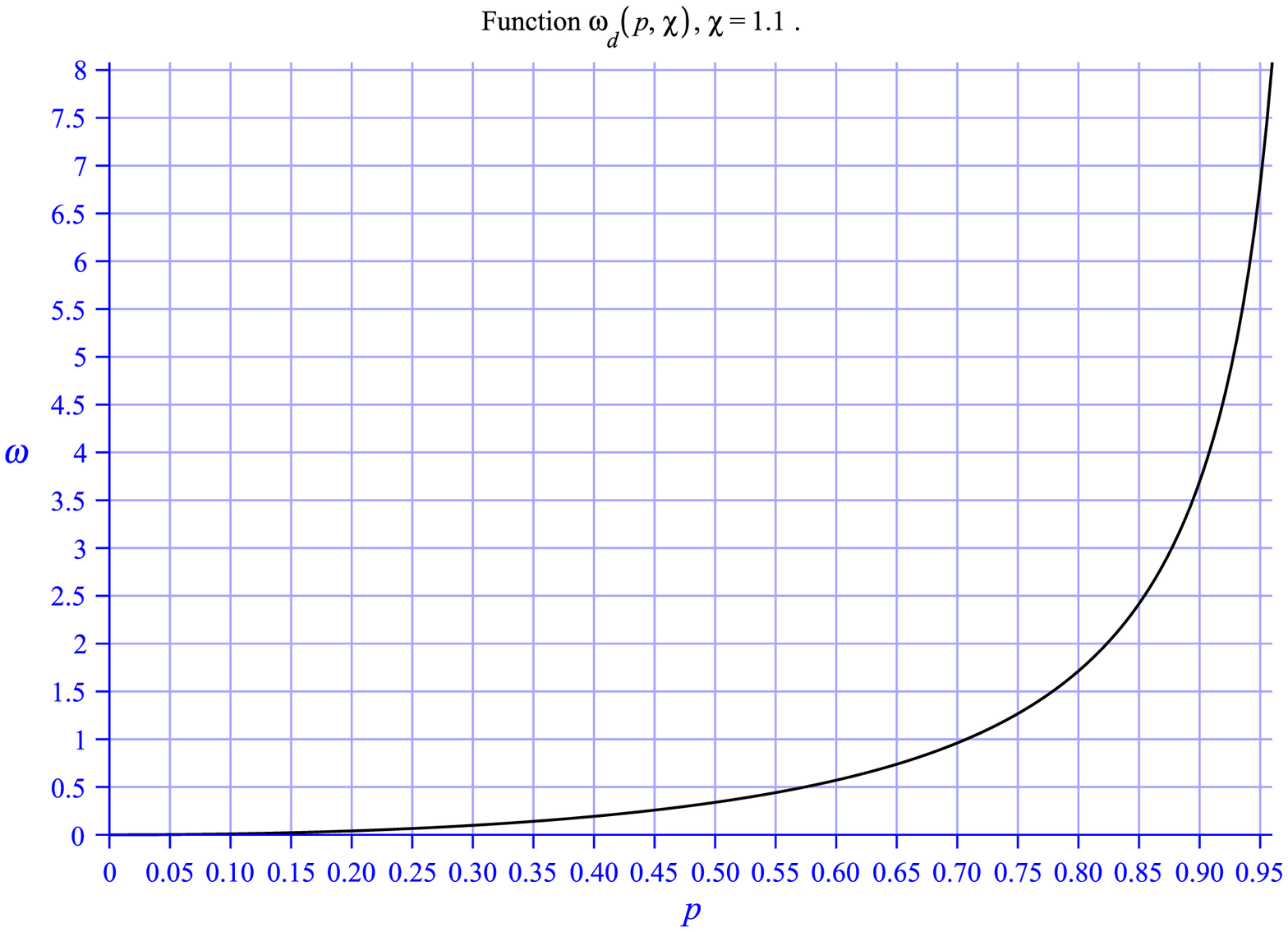}\hspace{0.8cm}\includegraphics[scale=0.35]{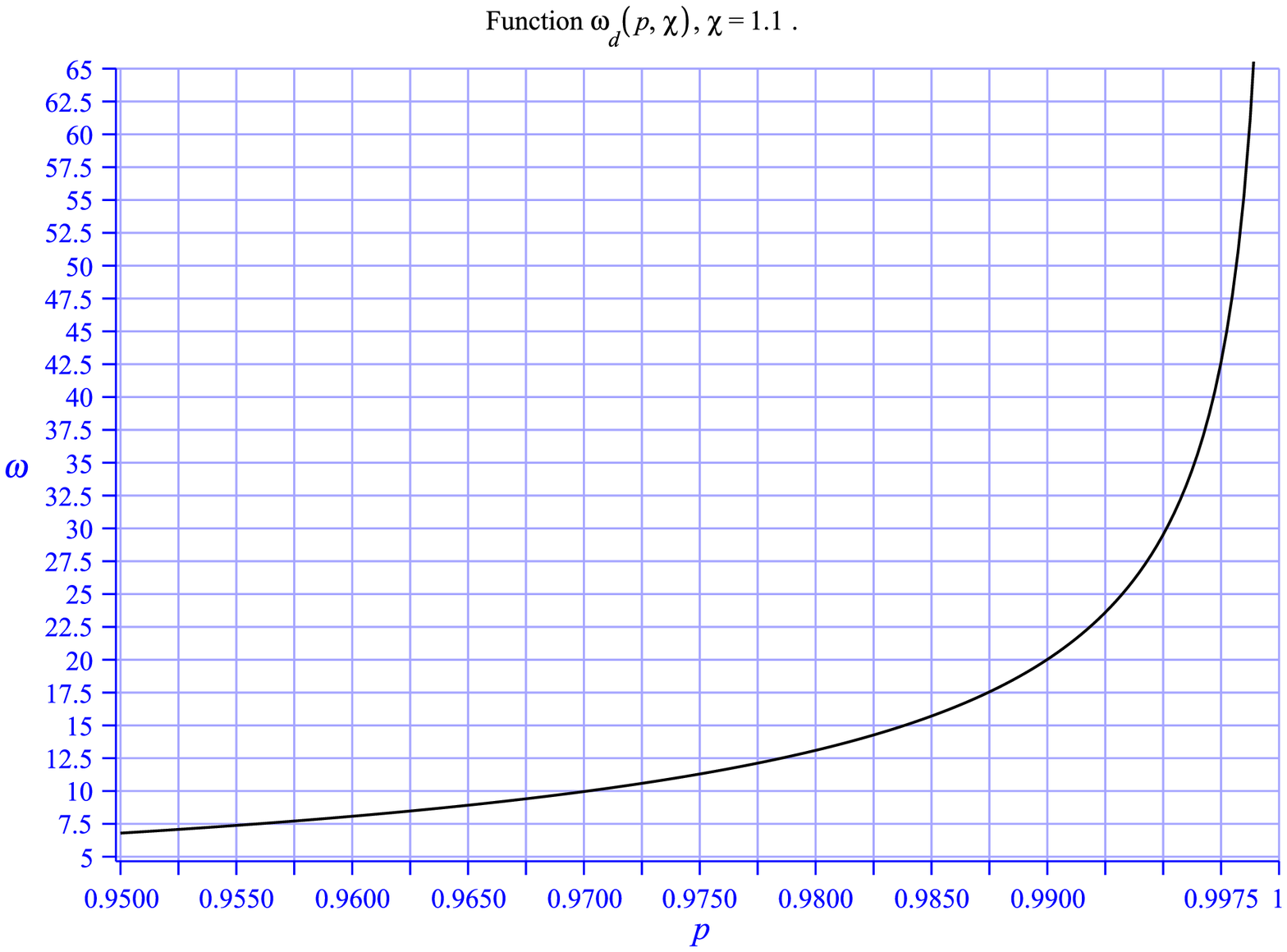}
\par\end{centering}
\centering{}(a)\hspace{6.5cm}(b)\caption{\label{fig:om-p-T1B1} Plots of function $\check{\omega}_{\mathrm{e}}\left(p,\chi\right)$
as in equation (\ref{eq:Dsup1d}) for $\chi=1.1$ and (a) $0<p<0.95$;
(b) $0.95<p<1$.}
\end{figure}

Using the EPD frequency $\check{\omega}_{\mathrm{e}}\left(p,\chi\right)$
defined by equations (\ref{eq:Dsup1d}) we obtain the following equivalent
form of the characteristic equation (\ref{eq:Dsup1a})
\begin{gather}
\mathscr{D}_{\mathrm{e}}\left(\check{u},p\right)=\mathscr{D}_{\mathrm{e}}\left(\check{u},p\right)-\frac{1}{\check{\omega}_{\mathrm{e}}^{2}\left(p,\chi\right)}=\frac{1}{\check{\omega}^{2}}-\frac{1}{\check{\omega}_{\mathrm{e}}^{2}\left(p,\chi\right)},\text{ where}\label{eq:Dsup2a}\\
\mathscr{D}_{\mathrm{e}}\left(\check{u},p\right)=\frac{\left(\check{u}-p\right)^{2}\left[\left({\it \chi}^{2}p+p^{3}-{\it \chi}^{2}\right)\check{u}^{2}+2{\it \chi}^{2}p\left(p-1\right)\check{u}-p^{2}{\it \chi}^{2}\right]}{\check{u}^{2}p^{4}\left(\check{u}^{2}-{\it \chi}^{2}\right)}.\nonumber 
\end{gather}
It is evident that for $\check{\omega}=\check{\omega}_{\mathrm{e}}$
velocity $\check{u}=p$ is a solution to equation (\ref{eq:Dsup2a})
of multiplicity $2$. We refer to equation (\ref{eq:Dsup2a}) as the
\emph{EPD form of the characteristic equation}. An algebraic factorized
form of rational function $\mathscr{D}_{\mathrm{e}}\left(\check{u},p\right)$
defined by equations (\ref{eq:Dsup2a}) is
\begin{equation}
\mathscr{D}_{\mathrm{e}}\left(\check{u},p\right)=\frac{\left({\it \chi}^{2}p+p^{3}-{\it \chi}^{2}\right)\left(\check{u}-p\right)^{2}\left(\check{u}-\lambda_{+}\left(p\right)\right)\left(\check{u}-\lambda_{+}\left(p\right)\right)}{\check{u}^{2}p^{4}\left(\check{u}^{2}-{\it \chi}^{2}\right)},\label{eq:Dsup2b}
\end{equation}
where
\begin{equation}
\lambda_{\pm}\left(p\right)=\frac{p{\it \chi}\left[\chi\left(1-p\right)\pm\sqrt{p^{3}-p\chi^{2}\left(1-p\right)}\right]}{p^{3}-\left(1-p\right)\chi^{2}}.\label{eq:Dsup2c}
\end{equation}

We infer based on the above analysis that for any given $p$ and $\chi$
the EPD values of the original TWT dimensionless parameters are
\begin{equation}
\text{EPD: }\check{u}=p,\quad\check{\gamma}=\check{\gamma}_{\mathrm{e}}\left(p,\chi\right)=\frac{\left(1-p\right)\left(p^{2}-{\it \chi}^{2}\right)^{2}}{p^{4}},\quad\check{\omega}=\check{\omega}_{\mathrm{e}}\left(p,\chi\right)=\frac{p^{2}}{\sqrt{\left(1-p\right)\left({\it \chi}^{2}-p^{3}\right)}}.\label{eq:Dsup2d}
\end{equation}

\subsection{Companion matrix spectral analysis\label{subsec:comp-mat}}

According to our review on matrix polynomials and associated with
them companion matrices in Section \ref{sec:mat-poly} the companion
matrix associated with matrix polynomial $\mathsf{M}_{u\omega}$ as
in equation (\ref{eq:CMuom1a}) takes the form
\begin{equation}
\mathscr{C}=\left[\begin{array}{rrrr}
0 & 0 & 1 & 0\\
0 & 0 & 0 & 1\\
\chi^{2} & \chi^{2} & 0 & 0\\
-\frac{\check{\gamma}}{\check{\omega}{}^{-2}-1} & -\frac{\check{\gamma}-1}{\check{\omega}{}^{-2}-1} & 0 & -\frac{2}{\check{\omega}{}^{-2}-1}
\end{array}\right].\label{eq:Compu1a}
\end{equation}
If $\check{u}$ is a characteristic velocity and consequently an eigenvalue
of companion matrix $\mathscr{C}$ then its unique eigenvector $Y\left(\check{u}\right)$
is defined by the following expression
\begin{equation}
Y\left(\check{u}\right)=\left[\begin{array}{r}
\chi^{2}\\
\check{u}{}^{2}-\chi^{2}\\
\chi^{2}\check{u}\\
\check{u}\left(\check{u}{}^{2}-\chi^{2}\right)
\end{array}\right].\label{eq:Compu1b}
\end{equation}
Notice that expression (\ref{eq:CMuom1a}) for the eigenvector $Y\left(\check{u}\right)$
of companion matrix $\mathscr{C}$ depends remarkably on parameter
$\chi$ and the value of the corresponding eigenvalue $\check{u}$
of $\mathscr{C}$ only.

The value $\mathscr{C}_{\mathrm{e}}$ of companion matrix $\mathscr{C}$
at the EPD point $\check{u}=p$ can be obtained by plugging in its
expression (\ref{eq:CMuom1a}) the EPD values of $\check{\gamma}=\check{\gamma}_{\mathrm{e}}$
and $\check{\omega}=\check{\omega}_{\mathrm{e}}$ defined in equations
(\ref{eq:Dsup2d}) resulting in
\begin{equation}
\mathscr{C}_{\mathrm{e}}=\left[\begin{array}{rrrr}
0 & 0 & 1 & 0\\
0 & 0 & 0 & 1\\
\chi^{2} & \chi^{2} & 0 & 0\\
\frac{\left(1-p\right)\left(p^{2}-{\it \chi}^{2}\right)^{2}}{p^{3}-\left(1-p\right)\chi^{2}} & \frac{\left(1-p\right){\it \chi}^{2}\left({\it \chi}^{2}-2p^{2}\right)-p^{5}}{p^{3}-\left(1-p\right)\chi^{2}} & 0 & \frac{2p^{4}}{p^{3}-\left(1-p\right)\chi^{2}}
\end{array}\right].\label{eq:Compu1c}
\end{equation}
Elementary but tedious analysis shows the Jordan canonical form $\mathscr{J}_{\mathrm{e}}$
of companion matrix $\mathscr{C}_{\mathrm{e}}$ is
\begin{equation}
\mathscr{J}_{\mathrm{e}}=\left[\begin{array}{rrrr}
p & 1 & 0 & 0\\
0 & p & 0 & 0\\
0 & 0 & \lambda_{+}\left(p\right) & 0\\
0 & 0 & 0 & \lambda_{-}\left(p\right)
\end{array}\right],\quad\lambda_{\pm}\left(p\right)=\frac{p{\it \chi}\left[\chi\left(1-p\right)\pm\sqrt{p^{3}-p\chi^{2}\left(1-p\right)}\right]}{p^{3}-\left(1-p\right)\chi^{2}}.\label{eq:Compu1d}
\end{equation}
Notice that expressions for eigenvalues $\lambda_{\pm}\left(p\right)$
are consistent with equations (\ref{eq:Dsup2b}) and (\ref{eq:Dsup2c})
as one may expect.

The Jordan basis of the generalized eigenspace of the EPD companion
matrix $\mathscr{C}_{\mathrm{e}}$ is formed by the eigenvector $Y_{\mathrm{e}}\left(p\right)$
and the so-called root vector $Y_{\mathrm{e}}^{\prime}\left(p\right)$
which are as follows
\begin{equation}
Y_{\mathrm{e}}\left(p\right)=\left[\begin{array}{r}
\chi^{2}\\
p^{2}-\chi^{2}\\
\chi^{2}p\\
p\left(p-\chi^{2}\right)
\end{array}\right],\quad Y_{\mathrm{e}}^{\prime}\left(p\right)=\partial_{p}Y_{\mathrm{e}}\left(p\right)=\left[\begin{array}{r}
0\\
2\check{u}\\
\chi^{2}\\
3\check{u}{}^{2}-\chi^{2}
\end{array}\right],\label{eq:Compu1e}
\end{equation}
where root vector $Y_{\mathrm{e}}^{\prime}\left(p\right)$ satisfies
the following relations
\begin{equation}
\left(\mathscr{C}_{\mathrm{e}}-p\mathbb{I}\right)Y_{\mathrm{e}}^{\prime}\left(p\right)=Y_{\mathrm{e}}\left(p\right),\quad\left(\mathscr{C}_{\mathrm{e}}-p\mathbb{I}\right)^{2}Y_{\mathrm{e}}^{\prime}\left(p\right)=0.\label{eq:Compu1f}
\end{equation}
The fact that $Y_{\mathrm{e}}^{\prime}\left(p\right)=\partial_{p}Y_{\mathrm{e}}\left(p\right)$
is of course not incidental. It can be argued based on the fact that
the eigenvector $Y\left(\check{u}\right)$ of companion matrix $\mathscr{C}$
defined by equation (\ref{eq:Compu1b}) depends on parameter $\chi$
and the value of the corresponding eigenvalue $\check{u}$ of $\mathscr{C}$
only. But regardless to the argument one can verify the validity of
equations (\ref{eq:Compu1f}) by tedious by straightforward evaluation.

Then using (i) the Jordan basis formed by vectors $Y_{\mathrm{e}}\left(p\right)$
and $Y_{\mathrm{e}}^{\prime}\left(p\right)$ defined by equations
(\ref{eq:Compu1e}) and (ii) the general expression (\ref{eq:Compu1b})
for the eigenvector $Y\left(\check{u}\right)$ of companion matrix
$\mathscr{C}$ we obtain complete Jordan basis for the EPD companion
matrix $\mathscr{C}_{\mathrm{e}}$ formed by columns of the following
matrix
\begin{equation}
\mathscr{Y}=\left[Y_{\mathrm{e}}\left(p\right),Y_{\mathrm{e}}^{\prime}\left(p\right),Y\left(\lambda_{+}\left(p\right)\right),Y\left(\lambda_{-}\left(p\right)\right)\right].\label{eq:Compu2a}
\end{equation}
Consequently, we have
\begin{equation}
\mathscr{C}_{\mathrm{e}}=\mathscr{Y}\mathscr{J}_{\mathrm{e}}\mathscr{Y}^{-1},\label{eq:Compu2b}
\end{equation}
where matrices $\mathscr{J}_{\mathrm{e}}$ and $\mathscr{Y}$ are
defined respectively by equations (\ref{eq:Compu1d}) and (\ref{eq:Compu2a}).

As to physical significance of the EPD phase velocity $p$ and the
EPD frequency $\check{\omega}_{\mathrm{e}}$ they can be detected
and identified by their intrinsic association with the onset of instability.
Indeed, gradually increasing frequency $\omega$ of probing excitation
of the TWT one can detect its value $\check{\omega}_{\mathrm{e}}$
when the instability sets up. At this point one can also assess the
value of the EPD phase velocity $p$ by comparing time dependent input
and output signals. Using then equations (\ref{eq:Dsup2d}) and assuming
that $p$ and $\check{\omega}_{\mathrm{e}}$ are known we recover
the TWT model parameters as follows:

\begin{equation}
{\it \chi}^{2}=p^{3}\left(1+\frac{p}{\left(1-p\right)\check{\omega}_{\mathrm{e}}}\right),\quad\check{\gamma}=\frac{\left[p^{2}-\check{\omega}_{\mathrm{e}}^{2}\left(1-p\right)^{2}\right]^{2}}{\left(1-p\right)\check{\omega}_{\mathrm{e}}^{4}}.\label{eq:Compu2c}
\end{equation}

\section{Using an EPD for enhanced sensing of small signals\label{sec:sensing}}

Based on our studies in Section \ref{sec:twt-epd} we develop here
an approach for using the EPD of the TWT for enhanced sensing of small
signals. This approach has some similarity to what we advanced in
\cite{FigPert} for simple circuits with EPDs but it is naturally
somewhat more complex. We remind that our primary motivation for considering
the TWT is that it can operate at much higher frequencies compare
to frequencies for lumped circuits.

Since TWTs are used mostly as amplifiers with the gain varying exponentially
with their length one might entertain an idea that the exponential
amplification can be exploited for sensing of small signals. The problem
with this idea though is that the origin of TWT amplification is an
instability and that is hardly compatible with enhanced sensing of
small signals. In addition to that, in the case of exponential amplification
a variety of noises that naturally occur in any TWT can obscure the
small sensor signal.

An EPD in a TWT is in fact also associated with an instability but
the EPD regime is at least marginally stable. In particular, if we
choose TWT regime to be near the EPD rather than exactly at it the
TWT operation can be stable as we showed in \cite{FigPert}. This
approach is of course a trade off allowing to buy the stability in
exchange for reduced value of the enhancement factor for the small
sensed signal. But even with such a trade off in place one can get
more than 100 fold enhancement \cite{FigPert}.

\subsection{Mathematical model for sensing}

We start with equations (\ref{eq:T1B1beta1bj}) and (\ref{eq:DsT1B1hj})
that relate the TWT system model parameters $\chi$ and $\check{\gamma}$
to the primary physical quantities, namely
\begin{equation}
\chi=\frac{w}{\mathring{v}}=\frac{1}{\mathring{v}\sqrt{CL}},\quad\check{\gamma}=\frac{\gamma}{\mathring{v}{}^{2}}=\frac{K_{\gamma}}{C},\quad K_{\gamma}=\frac{b^{2}}{\mathring{v}{}^{2}}\frac{e^{2}}{m}R_{\mathrm{sc}}^{2}\sigma_{\mathrm{B}}\mathring{n}.\label{eq:perchi1a}
\end{equation}
We proceed then with an assumption that the e-beam parameters $\mathring{v}$,
$\sigma_{\mathrm{B}}$ and $\mathring{n}$ are chosen, fixed and maintain
their values through the process of sensing. \emph{We suppose further
that: (i) it is either the distributed capacitance $C$ or the distributed
inductance $L$ utilized for sensing; (ii) parameter $C$ or $L$
which is selected for sensing is slightly altered by the small sensor
signal}. The relevant alteration caused by sensed signal is assessed
by measuring the relevant characteristic velocities $\check{u}$ of
the TWT. These velocities $\check{u}$ are solutions to the EPD form
(\ref{eq:Dsup2a}) of the characteristic equation, that is
\begin{equation}
\frac{\left(\check{u}-p\right)^{2}\left[\left({\it \chi}^{2}p+p^{3}-{\it \chi}^{2}\right)\check{u}^{2}+2{\it \chi}^{2}p\left(p-1\right)\check{u}-p^{2}{\it \chi}^{2}\right]}{\check{u}^{2}p^{4}\left(\check{u}^{2}-{\it \chi}^{2}\right)}=\frac{1}{\check{\omega}^{2}}-\frac{1}{\check{\omega}_{\mathrm{e}}^{2}\left(p,\chi\right)},\label{eq:perchi1b}
\end{equation}
where parameters $p$ and $\check{\omega}_{\mathrm{e}}$ are related
to parameters $\check{\gamma}$ and ${\it \chi}$ by equations (\ref{eq:Compu2c}).
More precisely, with the equation (\ref{eq:perchi1b}) in mind we
probe the TWT at frequency $\check{\omega}$ of our choosing and then
measure the phase velocities $\check{u}$ of the excited eigenmodes
of the TWT system. Since according to our TWT model these velocities
$\check{u}$ satisfy characteristic equation (\ref{eq:perchi1b})
we can relate them to values of parameters $p$ and ${\it \chi}$.
Having found $p$ and ${\it \chi}$ we can recover then the values
of $C$ and $L$ based on equations (\ref{eq:Compu2c}) and (\ref{eq:perchi1a})
as we show below.

To have a clarity on what kind of the TWT states are considered. We
remind that the TWT significant properties are encoded in its companion
matrix $\mathscr{C}$ defined by equation (\ref{eq:Compu1a}). This
matrix in turn is determined by parameters $\check{\gamma}$, $\chi$
and frequency $\check{\omega}$. Instead of parameter $\check{\gamma}$
we can use the nodal velocity $p$ related to it by equation (\ref{eq:gamdp1a}).
Then based on the mentioned quantities we introduce the following
definitions of the TWT configuration and TWT state.
\begin{defn}[TWT configuration and state]
 \emph{TWT configuration} is defined as a pair of two dimensionless
TWT parameters: (i) $\chi=\frac{w}{\mathring{v}}$; (ii) the nodal
velocity $p=p\left(\chi,\check{\gamma}\right)$ satisfying equation
(\ref{eq:Dsup2d}). \emph{TWT state} is defined as a triple of o dimensionless
parameters: $\chi$, $p$ and $\check{\omega}$. Frequency $\check{\omega}$
can be viewed as a parameter that selects the corresponding four TWT
eigenmodes associated with the TWT companion matrix $\mathscr{C}$
defined by equation (\ref{eq:Compu1a}). The selection is facilitated
physically by exciting/probing the TWT at frequency $\check{\omega}$.
An \emph{EPD state} is defined by a particular choice of its parameters,
namely $\chi$, $p$ and $\check{\omega}=\check{\omega}_{\mathrm{e}}\left(p,\chi\right)$
where the EPD frequency $\check{\omega}_{\mathrm{e}}\left(p,\chi\right)$
is defined by equations (\ref{eq:Dsup1d}). \emph{Work point state}
is defined by the following choice of its parameters $\chi$, $p$
and $\check{\omega}=\check{\omega}_{\mathrm{w}}<\check{\omega}_{\mathrm{e}}\left(p,\chi\right)$.
We refer to $\check{\omega}_{\mathrm{w}}$ as \emph{work point frequency}.
There is a flexibility in choosing $\check{\omega}_{\mathrm{w}}$
to be proximate to the EPD frequency $\check{\omega}_{\mathrm{e}}\left(p,\chi\right)$
when at the same time to maintain certain distance from it to provide
for the stability of the TWT operation as explained in Remark \ref{rem:workpoint}.
There are exactly two eigenmodes with their phase velocities $\check{u}$
close to $p$ that are of particular significance.
\end{defn}

\begin{rem}[work point]
\label{rem:workpoint} It turns out that the ideal EPD state is intrinsically
only marginally stable and the purpose of the work point is to overcome
this problem. The stability of the work point is achieved by making
a deliberate small departure from the EPD frequency $\check{\omega}_{\mathrm{e}}$.
The departure is achieved by appropriate selection of the probing
frequency $\check{\omega}=\check{\omega}_{\mathrm{w}}$ so that $\check{\omega}_{\mathrm{w}}<\check{\omega}_{\mathrm{e}}$.
An additional benefit and utility of the work point frequency $\check{\omega}_{\mathrm{w}}$
is that it lifts the characteristic velocity degeneracy causing the
velocity split. This velocity split can be measured and used to determine
the small sensed signal, see Theorem \ref{thm:velsplit}, Fig. \ref{fig:approx-EPD}
and equations (\ref{eq:delCLom1b}) and (\ref{eq:delCLom1c}).
\end{rem}

Suppose that the TWT configuration before sensing is defined by parameters
$\chi$ and $p$ and the TWT configuration altered by the small sensed
signal is defined by parameters $\chi^{\prime}$ and $p^{\prime}$.
We proceed then with introducing a larger set of parameters associated
with the EPD state of the TWT before it it receives the small sensed
signal:

\begin{equation}
\text{EPD state: }C,\quad L,\quad p,\quad\chi,\quad\check{\omega}=\check{\omega}_{\mathrm{e}}\left(p,\chi\right).\label{eq:perchi1c}
\end{equation}
Using the above EPD state as a reference point we introduce a larger
set of parameters for \emph{altered EPD state} associated with the
small sensed signal:
\begin{gather}
\text{altered EPD state: }C^{\prime}=C\left(1+\delta_{C}\right),\quad L{}^{\prime}=L\left(1+\delta_{L}\right),\quad p^{\prime}=p\left(1+\delta_{p}\right),\label{eq:perchi1ca}\\
\check{\left(\chi^{2}\right)^{\prime}=\chi^{2}\left(1+\delta_{\chi^{2}}\right),\quad\omega}=\check{\omega}_{\mathrm{e}}^{\prime}=\check{\omega}_{\mathrm{e}}\left(p^{\prime},\chi{}^{\prime}\right)=\check{\omega}_{\mathrm{e}}\left(1+\delta_{\omega}\right),\quad\delta_{\omega}=\frac{\check{\omega}_{\mathrm{e}}^{\prime}-\check{\omega}_{\mathrm{e}}}{\check{\omega}_{\mathrm{e}}},\nonumber 
\end{gather}
where as a matter of computational convenience we use parameter $\chi^{2}$
rather than $\chi$. In relations (\ref{eq:perchi1c})-(\ref{eq:perchi1e})
the relative variation coefficients $\delta_{*}$ are assumed to be
small and satisfy the following relations
\begin{equation}
\left|\delta_{C}\right|,\:\left|\delta_{L}\right|,\:\left|\delta_{\chi^{2}}\right|,\:\left|\delta_{p}\right|\ll1;\quad\left|\delta_{\omega}\right|\ll\left|\delta_{\mathrm{w}}\right|\ll1.\label{eq:perchi1f}
\end{equation}

Using once again the EPD state as a reference point we introduce parameters
of the work point state and its altered version as follows:

\begin{gather}
\text{work point state: }C,\quad L,\quad p,\quad\chi,\quad\check{\omega}=\check{\omega}_{\mathrm{w}}=\check{\omega}_{\mathrm{e}}\left(1+\delta_{\mathrm{w}}\right),\quad\delta_{\mathrm{w}}=\frac{\check{\omega}_{\mathrm{w}}-\check{\omega}_{\mathrm{e}}}{\check{\omega}_{\mathrm{e}}},\label{eq:perchi1d}
\end{gather}

\begin{gather}
\text{altered work point state: }C^{\prime}=C\left(1+\delta_{C}\right),\quad L{}^{\prime}=L\left(1+\delta_{L}\right),\quad p^{\prime}=p\left(1+\delta_{p}\right),\label{eq:perchi1e}\\
\left(\chi^{2}\right)^{\prime}=\chi^{2}\left(1+\delta_{\chi^{2}}\right),\;\check{\omega}=\check{\omega}_{\mathrm{w}}=\check{\omega}_{\mathrm{e}}\left(1+\delta_{\mathrm{w}}\right),\;\check{\omega}_{\mathrm{e}}^{\prime}=\check{\omega}_{\mathrm{e}}\left(p^{\prime},\chi{}^{\prime}\right)=\check{\omega}_{\mathrm{e}}\left(1+\delta_{\omega}\right),\;\delta_{\omega}=\frac{\check{\omega}_{\mathrm{e}}^{\prime}-\check{\omega}_{\mathrm{e}}}{\check{\omega}_{\mathrm{e}}}.\nonumber 
\end{gather}
\emph{Relative incremental frequency differences $\delta_{\mathrm{w}}$
and $\delta_{\omega}$ employ effectively the EPD frequency $\check{\omega}_{\mathrm{e}}$
as a natural frequency unit.}

Notice that parameters in relations (\ref{eq:perchi1c})-(\ref{eq:perchi1e})
are not independent but rather parameters $\chi$, $p$ and consequently
$\check{\omega}_{\mathrm{e}}$ can be expressed in terms of the primary
physical parameters $C$ and $L$. Namely, according to relations
(\ref{eq:Dsup2d}) and (\ref{eq:perchi1a}) we have
\begin{equation}
\chi=\frac{1}{\mathring{v}\sqrt{CL}},\quad\check{\gamma}=\frac{K_{\gamma}}{C}=\frac{\left(1-p\right)\left(p^{2}-{\it \chi}^{2}\right)^{2}}{p^{4}},\quad\check{\omega}_{\mathrm{e}}=\frac{p^{2}}{\sqrt{\left(1-p\right)\left({\it \chi}^{2}-p^{3}\right)}}.\label{eq:perchi2a}
\end{equation}
Since in view of made assumptions the e-beam parameters $\mathring{v}$
and $K_{\gamma}$ can be considered to be constants the first two
equations in (\ref{eq:perchi2a}) allow to express $\chi$ and $p$
as a functions of $C$ and $L$. Consequently the third equation in
(\ref{eq:perchi2a}) determines $\check{\omega}_{\mathrm{e}}$ as
a function of $C$ and $L$ also.

We proceed now with the derivation of linear approximations, that
is the differentials, for $\delta_{\chi^{2}}$, $\delta_{p}$ and
$\delta_{\omega}$ in terms of $\delta_{C}$ and $\delta_{L}$. The
assumed smallness of $\delta_{C}$, $\delta_{L}$ and $\delta_{\omega}$
imply the smallness of $\delta_{\chi^{2}}$, $\delta_{p}$ and $\delta_{\omega}^{\prime}$.
Hence \emph{the work point and the altered work point states are all
close to the EPD} state. Using the proximity of all these states we
obtain the following first order approximation to the characteristic
equation (\ref{eq:perchi1b}) for the altered state
\begin{gather}
\left(\check{u}-p^{\prime}\right)^{2}\cong S\left(\check{\omega}_{\mathrm{e}}^{\prime}-\check{\omega}_{\mathrm{w}}\right)=S\check{\omega}_{\mathrm{e}}\left(\delta_{\omega}-\delta_{\mathrm{w}}\right),\label{eq:upS1a}
\end{gather}
where
\begin{equation}
\delta_{\omega}=\frac{\check{\omega}_{\mathrm{e}}^{\prime}-\check{\omega}_{\mathrm{e}}}{\check{\omega}_{\mathrm{e}}},\quad\delta_{\mathrm{w}}=\frac{\check{\omega}_{\mathrm{w}}-\check{\omega}_{\mathrm{e}}}{\check{\omega}_{\mathrm{e}}},\quad\check{\omega}_{\mathrm{e}}=\check{\omega}_{\mathrm{e}}\left(p,\chi\right),\quad\check{\omega}_{\mathrm{e}}^{\prime}=\check{\omega}_{\mathrm{e}}\left(p^{\prime},\chi{}^{\prime}\right),\label{eq:upS1aa}
\end{equation}
\begin{equation}
S=S\left(p,\chi\right)=\frac{2p^{4}\left({\it \chi}^{2}-p^{2}\right)\left[\left(1-p\right)\left({\it \chi}^{2}-p^{3}\right)\right]^{\frac{3}{2}}}{p^{2}\left(4{\it \chi}^{2}-3{\it \chi}^{2}p-p^{3}\right)}>0,\quad0<p<1,\quad\chi>1.\label{eq:upS1ab}
\end{equation}
We refer to equation (\ref{eq:upS1a}) as \emph{EPD approximation
to the characteristic equation}, and we refer to quantities $\delta_{\omega}$
and $\delta_{\mathrm{w}}$ $\Delta^{\prime}$ and $\Delta_{w}$ respectively
as \emph{relative EPD increment} and \emph{relative work point increment}.
\emph{The EPD approximation to the characteristic equation which is
a quadratic equation selects two characteristic velocities out of
the total of four by their property to be proximate to the EPD velocity
$p$.} Equations (\ref{eq:upS1a}) readily imply the following statement.
\begin{thm}[velocity split near the EPD]
\label{thm:velsplit}. The velocities $\check{u}_{\pm}$ that solve
the quadratic in $\check{u}$ equation (\ref{eq:upS1a}) are
\begin{equation}
\check{u}_{\pm}=p^{\prime}\pm\sqrt{S\left(\check{\omega}_{\mathrm{e}}^{\prime}-\check{\omega}_{\mathrm{w}}\right)}=p^{\prime}\pm\sqrt{S\check{\omega}_{\mathrm{e}}\left(\delta_{\omega}-\delta_{\mathrm{w}}\right)},\label{eq:upS1b}
\end{equation}
where factor $S>0$ satisfies relations (\ref{eq:upS1ab}) and frequency
$\check{\omega}_{\mathrm{e}}>0$ satisfies the last equation in (\ref{eq:perchi2a}).
The parameters $\chi$ and $p$ that determine $S$ and $\check{\omega}_{\mathrm{e}}$
are associated with the EPD state as in (\ref{eq:perchi1c}). Since
$S>0$ velocities $\check{u}_{\pm}$ defined by equation (\ref{eq:upS1b})
are real if and only if
\begin{equation}
\check{\omega}_{\mathrm{w}}<\check{\omega}_{\mathrm{e}}^{\prime}.\label{eq:upS1ba}
\end{equation}
Equations (\ref{eq:upS1b}) imply the following representation for
the velocity split $\check{u}_{+}-\check{u}_{-}$ at the EPD:
\begin{gather}
\check{u}_{+}-\check{u}_{-}=2\sqrt{S\left(\check{\omega}_{\mathrm{e}}^{\prime}-\check{\omega}_{\mathrm{w}}\right)}=2\sqrt{S\check{\omega}_{\mathrm{e}}\left(\delta_{\omega}-\delta_{\mathrm{w}}\right)},\label{eq:upS1c}\\
\delta_{\omega}=\frac{\check{\omega}_{\mathrm{e}}^{\prime}-\check{\omega}_{\mathrm{e}}}{\check{\omega}_{\mathrm{e}}},\quad\delta_{\mathrm{w}}=\frac{\check{\omega}_{\mathrm{w}}-\check{\omega}_{\mathrm{e}}}{\check{\omega}_{\mathrm{e}}},\nonumber 
\end{gather}
where frequencies $\check{\omega}_{\mathrm{e}}$, $\check{\omega}_{\mathrm{e}}^{\prime}$
and $\check{\omega}_{\mathrm{w}}$ are associated respectively with
the EPD state, the altered state and the work point state defined
by equations (\ref{eq:perchi1c}), (\ref{eq:perchi1d}) and (\ref{eq:perchi1e}).

If the following inequalities hold
\begin{equation}
\left|\delta_{\omega}\right|<\left|\delta_{\mathrm{w}}\right|\text{ and }-\delta_{\mathrm{w}}>0,\label{eq:upS1ca}
\end{equation}
then
\begin{equation}
\delta_{\omega}-\delta_{\mathrm{w}}=\frac{\check{\omega}_{\mathrm{e}}^{\prime}-\check{\omega}_{\mathrm{w}}}{\check{\omega}_{\mathrm{e}}}>0\label{eq:upS1cb}
\end{equation}
readily implying the inequality (\ref{eq:upS1ba}).
\end{thm}

\begin{rem}[stability and square root enhancement factor]
\label{rem:sqrt-enh} We refer to inequalities
\begin{equation}
\check{\omega}_{\mathrm{w}}<\check{\omega}_{\mathrm{e}},\;\check{\omega}_{\mathrm{e}}^{\prime},\label{eq:upS1cc}
\end{equation}
that appear in the statement of Theorem \ref{thm:velsplit} as the\emph{
frequency stability conditions}. These conditions are relevant to
the stability for they imply the real-valuedness of the velocities
$\check{u}_{\pm}$. Note that inequalities (\ref{eq:upS1ca}) imply
frequency stability conditions (\ref{eq:upS1cc}). 

Note also that according to equation (\ref{eq:upS1c}) the velocity
split $\check{u}_{+}-\check{u}_{-}$, a quantity that can be measured,
is proportional to $\sqrt{\check{\omega}_{\mathrm{e}}^{\prime}-\check{\omega}_{\mathrm{w}}}$.
The square root operation applied to the small frequency difference
$\check{\omega}_{\mathrm{e}}^{\prime}-\check{\omega}_{\mathrm{w}}$
effectively ``magnifies'' it. This is a typical manifestation of
the proximity to the EPD. \emph{The rise of the square root operation
can be traced to the EPD approximation to the characteristic equation
(\ref{eq:upS1a}) which is a quadratic in $\check{u}$ equation}.
\end{rem}

The EPD approximation to the characteristic equation (\ref{eq:upS1a})
and its solutions $\check{u}_{\pm}$ are illustrated by Fig. \ref{fig:approx-EPD}.
\begin{figure}
\begin{centering}
\includegraphics[scale=0.4]{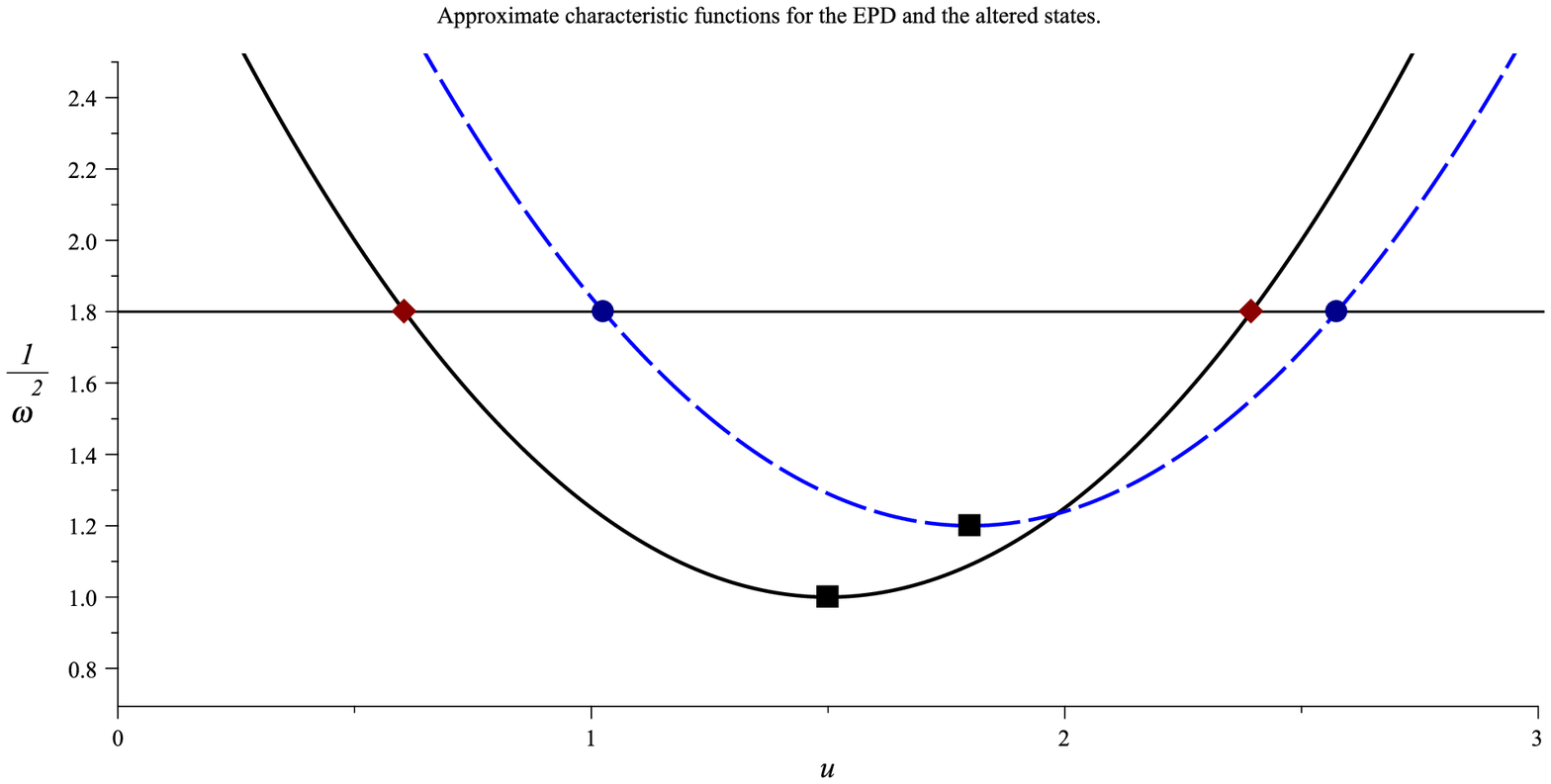}
\par\end{centering}
\centering{}\caption{\label{fig:approx-EPD} An illustrative plot for the approximate characteristic
equation similar (\ref{eq:upS1a}) for EPD state represented by black
solid parabola and the altered state represented by blue dashed parabola
(horizontal axis \textendash{} $u$, vertical axis \textendash{} $\frac{1}{\check{\omega}^{2}}$).
The black horizontal line corresponds to the work frequency $\check{\omega}=\check{\omega}_{\mathrm{w}}$.
The points of intersection of the horizontal line and blue dashed
parabola shown as blue disk points represent solutions $\check{u}_{\pm}$
to equation (\ref{eq:upS1b}). The points of intersection of the horizontal
line and black solid parabola are shown as brown diamond points represents
solutions $\check{u}_{\pm}$ to equation similar (\ref{eq:upS1b})
corresponding to $p$ rather than $p^{\prime}$. The EPDs correspond
the to vertexes of the parabolas shown as black square points. Notice
that the EPD frequencies $\check{\omega}_{\mathrm{e}}$ and $\check{\omega}_{\mathrm{e}}^{\prime}$
associated with the parabolas vertices satisfy the frequency stability
conditions (\ref{eq:upS1cc}).}
\end{figure}

\subsection{Sensing algorithm}

Our approach to utilize the TWT for sensing is as follows. First of
all, we assume that (i) the small sensed signal alters either distributed
capacitance $C$ or distributed inductance $L$; (ii) the physical
alteration of the TWT indicated in (i) is described mathematically
by relations (\ref{eq:perchi1c})-(\ref{eq:perchi1f}). Second of
all, we assume that prior to any measurements the values of parameters
$C$, $L$, $\chi$ and $p$ that determine the current state of the
TWT are established and known. These values constitute a reference
point for an assessment of the results of sensing. Third of all, we
assume that the relationship between $\delta_{C}$ or $\delta_{L}$
and the small sensed signal is known and is an integral part of the
sensor.

\emph{The consecutive steps of our sensing algorithm are as follows:}
\begin{enumerate}
\item Apply the small sensed signal to the relevant component of the TWT.
Probe then the TWT by an excitation at the work point frequency $\check{\omega}=\check{\omega}_{\mathrm{w}}$. 
\item Observe the two TWT eigenmodes and assess the corresponding to them
characteristic velocities $\check{u}_{+}$ and $\check{u}_{-}$ that
are close to the EPD velocity $p$ and satisfy the EPD approximation
(\ref{eq:upS1a}) to the characteristic equation.
\item Compute the velocity split $\check{u}_{+}-\check{u}_{-}$. Depending
on what of the parameters $C$ or $L$ was utilized for sensing find
the corresponding value $\delta_{C}$ or $\delta_{L}$ using respectively
formulas (\ref{eq:delCLom1b}) and (\ref{eq:delCLom1c}) below.
\item Based on assumed to be known relationship between $\delta_{C}$ or
$\delta_{L}$ and the small sensed signal recover the value of the
small sensed signal.
\end{enumerate}
Fig. \ref{fig:approx-EPD} illustrates graphically the proposed sensing
approach. The sensing approach utilizing the EPD of the TWT is conceptually
similar to the one for circuits advanced in \cite{FigPert} but it
is naturally more complex. The reason for complexity is that the TWT
as a physical system is naturally a more complex system compare to
the simple circuits with EPDs we advanced in \cite{FigPert}.

In the light of Theorem \ref{thm:velsplit} and Remark \ref{rem:sqrt-enh}
let us take a closer look at the relative EPD increment $\delta_{\omega}$
and the relative work point increment $\delta_{\mathrm{w}}$ defined
by equations (\ref{eq:upS1aa}). Both relative increments $\delta_{\omega}$
and $\delta_{\mathrm{w}}$ in equations (\ref{eq:upS1aa}) use the
EPD frequency $\check{\omega}_{\mathrm{e}}=\check{\omega}_{\mathrm{e}}\left(p,\chi\right)$
as a reference point and effectively as a natural frequency unit.
The stability of the TWT operation requires the characteristic velocities
$\check{u}_{\pm}$ defined by equations (\ref{eq:upS1b}) to be real-valued.
This requirement is fulfilled always if and only if the expression
under square root in the right-hand side of equation (\ref{eq:upS1a})
is non-negative. That in turn leads to
\begin{equation}
\delta_{\omega}-\delta_{\mathrm{w}}=\frac{\check{\omega}_{\mathrm{e}}^{\prime}-\check{\omega}_{\mathrm{w}}}{\check{\omega}_{\mathrm{e}}}>0,\text{ implying }\check{\omega}_{\mathrm{w}}<\check{\omega}_{\mathrm{e}},\label{eq:upS1d}
\end{equation}
since $\delta_{\omega}$ can be zero and factors $\check{\omega}_{\mathrm{e}},S>0$
according to relations (\ref{eq:perchi2a}) and (\ref{eq:upS1ab}).
The second inequality in (\ref{eq:upS1d}) is one of the frequency
stability conditions (\ref{eq:upS1cc}). As it is indicated previously
in the last relations in (\ref{eq:perchi1f}) we require the following
relations to hold
\begin{equation}
\left|\delta_{\omega}\right|=\left|\frac{\check{\omega}_{\mathrm{e}}^{\prime}-\check{\omega}_{\mathrm{e}}}{\check{\omega}_{\mathrm{e}}}\right|\ll\left|\delta_{\mathrm{w}}\right|=\left|\frac{\check{\omega}_{\mathrm{w}}-\check{\omega}_{\mathrm{e}}}{\check{\omega}_{\mathrm{e}}}\right|\ll1.\label{eq:upS1e}
\end{equation}
The point of inequalities (\ref{eq:upS1e}) as a requirement is to
assure that the altered EPD frequency $\check{\omega}_{\mathrm{e}}^{\prime}$
satisfies the frequency stability conditions (\ref{eq:upS1cc}) as
soon as the EPD frequency $\check{\omega}_{\mathrm{e}}$ satisfies
them. Consequently relations (\ref{eq:upS1e}) assure the effectiveness
and robustness of described above approach to sensing.

We proceed now with relating the relative increments $\delta_{\omega}$
and $\delta_{\mathrm{w}}$ defined by equations (\ref{eq:upS1aa})
to quantities $\delta_{C}$ and $\delta_{L}$ that are directly effected
by the small sensor signal. Using equations (\ref{eq:perchi2a}) and
assuming that $\delta_{C}$ and $\delta_{L}$ are small after tedious
but elementary evaluations we find the following first order approximations
\begin{equation}
\delta_{\chi^{2}}\cong-\left(\delta_{C}+\delta_{L}\right),\quad\delta_{p}\cong-\frac{\left(1-p\right)\left[\left(\chi^{2}+p^{2}\right)\delta_{{\it C}}+2\chi^{2}\delta_{{\it L}}\right]}{\left(4-3p\right)\chi^{2}-p^{3}}.\label{eq:upS2a}
\end{equation}
Using representation (\ref{eq:perchi2a}) for $\check{\omega}_{\mathrm{e}}\left(p,\chi\right)$
and equations (\ref{eq:upS2a}) (\ref{eq:upS2a}) under assumption
that $\delta_{C}$ and $\delta_{L}$ are small we obtain the following
first order approximation
\begin{equation}
\delta_{\omega}=\frac{\check{\omega}_{\mathrm{e}}\left(p^{\prime},\chi{}^{\prime}\right)-\check{\omega}_{\mathrm{e}}\left(p,\chi\right)}{\check{\omega}_{\mathrm{e}}\left(p,\chi\right)}\cong-\frac{p^{2}\delta_{C}+{\it \chi}^{2}\delta_{L}}{2\left({\it \chi}^{2}-p^{3}\right)}.\label{eq:upS2d}
\end{equation}
Fig. \ref{fig:incr-om-C} is a graphical representations of function
the increment $\check{\omega}_{\mathrm{e}}\left(p^{\prime},\chi{}^{\prime}\right)-\check{\omega}_{\mathrm{e}}\left(p,\chi\right)$
and its first order approximation as in equation (\ref{eq:upS2d})
for ranges of values of relative variation coefficient $\delta_{C}$
when $\delta_{L}=0$. Fig. \ref{fig:incr-om-L} similarly is a graphical
representation of function the increment $\check{\omega}_{\mathrm{e}}\left(p^{\prime},\chi{}^{\prime}\right)-\check{\omega}_{\mathrm{e}}\left(p,\chi\right)$
and its first order approximation as in equation (\ref{eq:upS2d})
for two ranges of values of the relative variation coefficient $\delta_{L}$
when $\delta_{C}=0$. Figures \ref{fig:incr-om-C} and \ref{fig:incr-om-L}
indicate that for the chosen values of parameters when the value of
relative variation coefficient $\delta_{C}$ or $\delta_{L}$ is under
$5\%$ the simple formula (\ref{eq:upS2d}) yields pretty accurate
values compare to the exact value of the increment. If values relative
variation coefficient $\delta_{C}$ or $\delta_{L}$ values are below
$1\%$ the approximation values are nearly the same as the exact values.
\begin{figure}[h]
\begin{centering}
\hspace{-1cm}\includegraphics[scale=0.28]{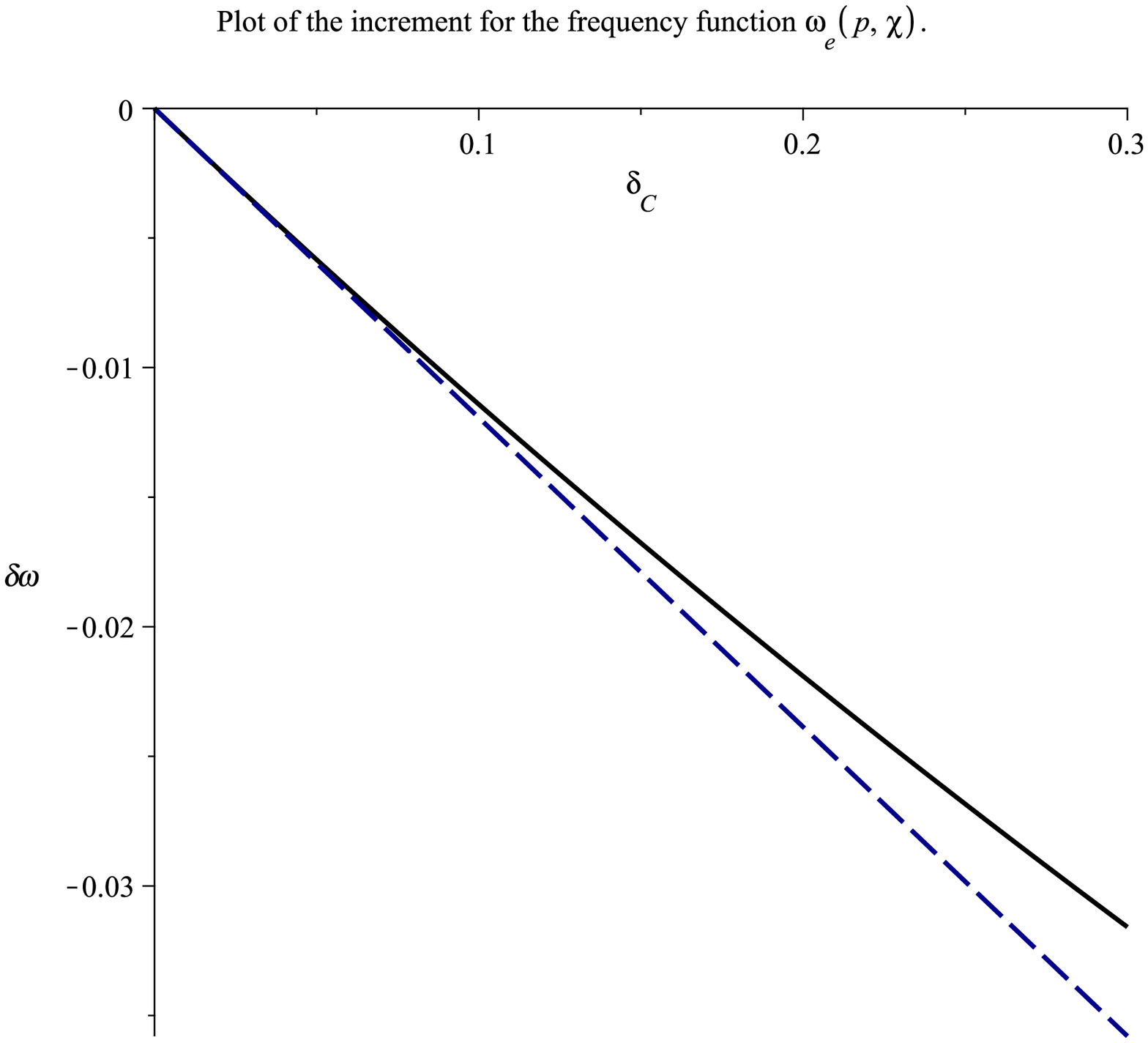}\hspace{2cm}\includegraphics[scale=0.28]{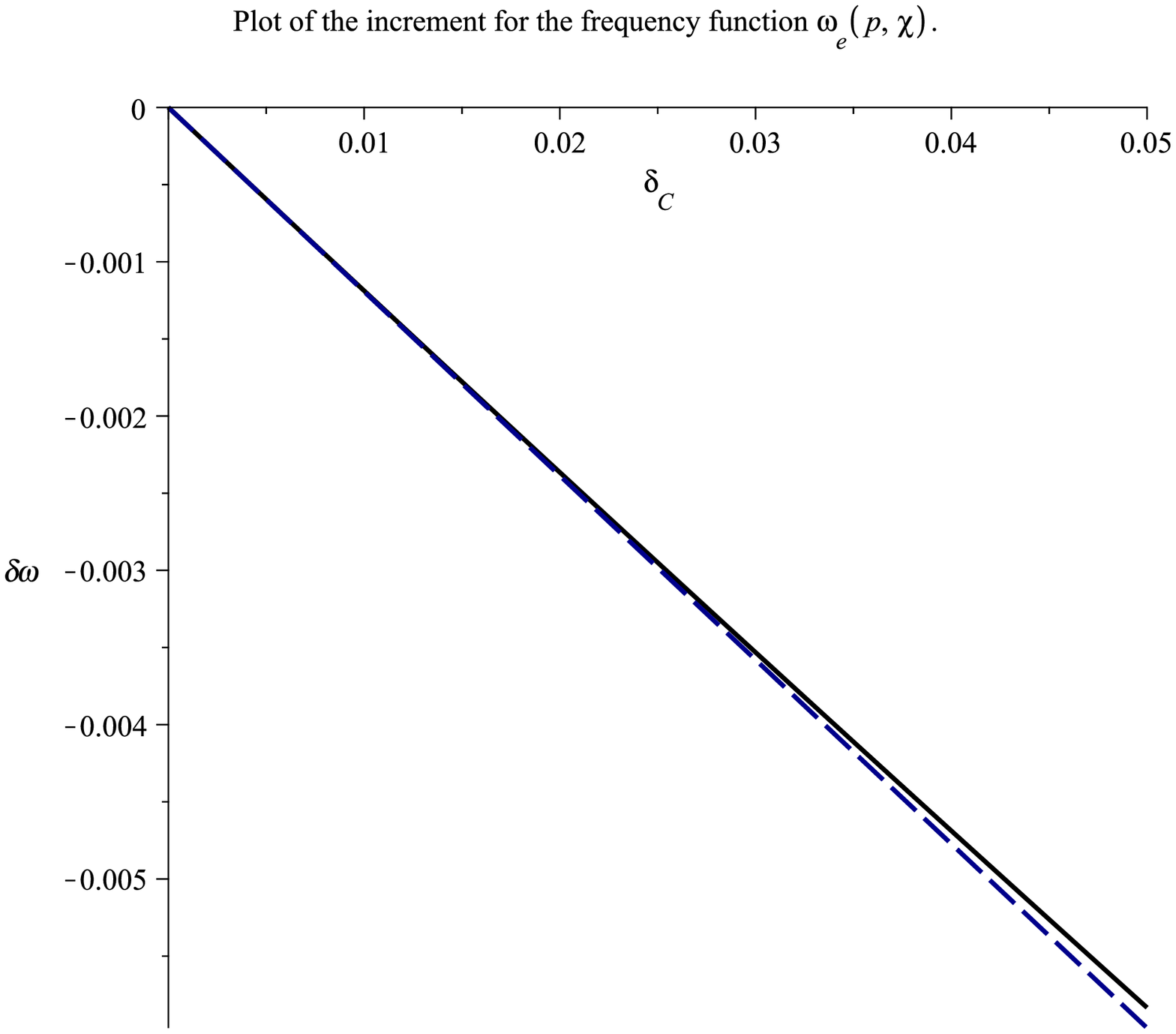}
\par\end{centering}
\centering{}(a)\hspace{7cm}(b)\caption{\label{fig:incr-om-C} Plots of the increment $\check{\omega}_{\mathrm{e}}\left(p^{\prime},\chi{}^{\prime}\right)-\check{\omega}_{\mathrm{e}}\left(p,\chi\right)$
and its first order approximation as in equation (\ref{eq:upS2d})
for $\mathring{v}=1.1$, $K_{\gamma}=1.2$, $C=1.3$, $L=1.4$, $\delta_{L}=0$
and: (a) $0<\delta_{C}<0.3$; (b) $0<\delta_{C}<0.05$.}
\end{figure}
 
\begin{figure}[h]
\begin{centering}
\hspace{-0.5cm}\includegraphics[scale=0.28]{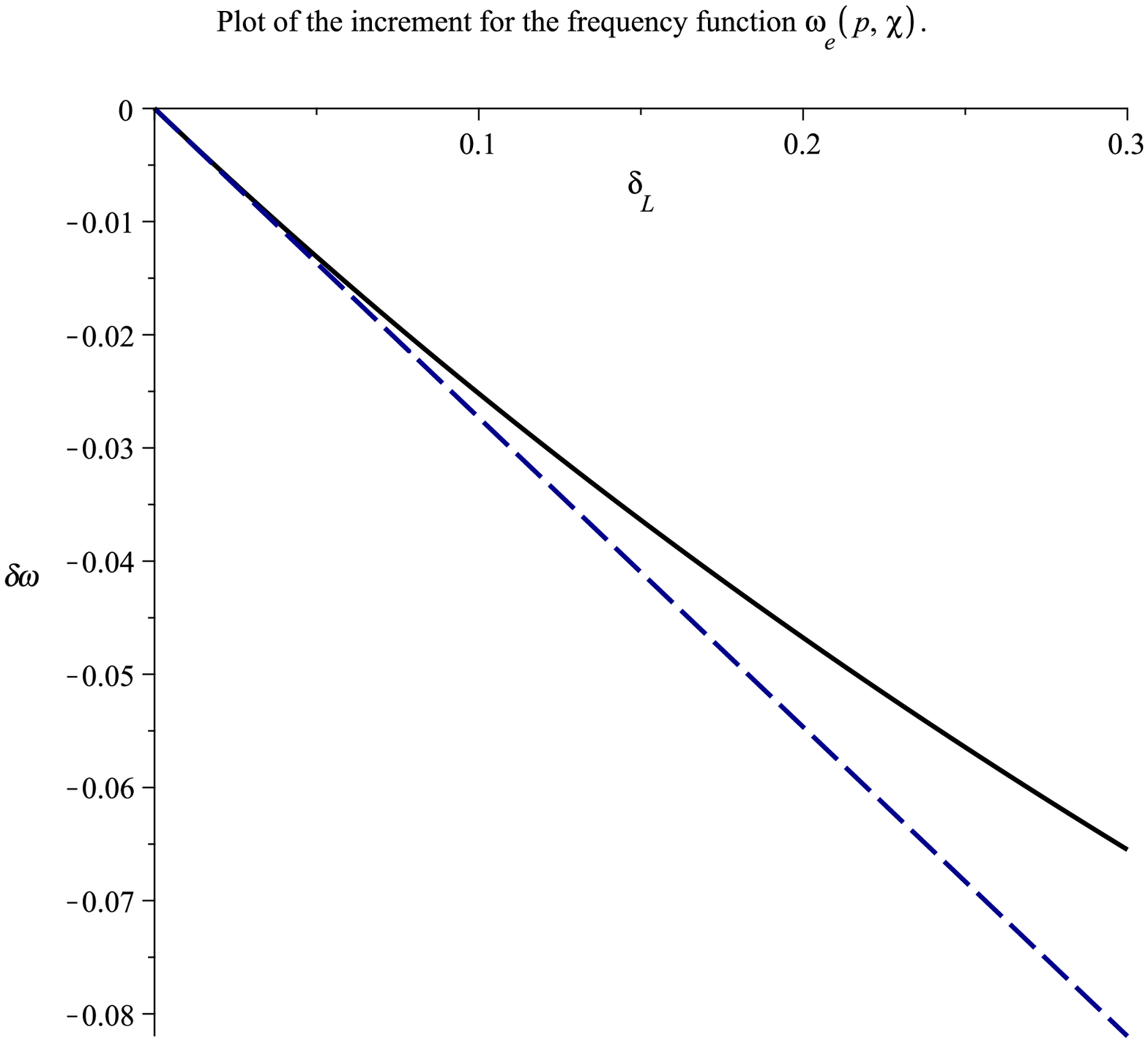}\hspace{2cm}\includegraphics[scale=0.3]{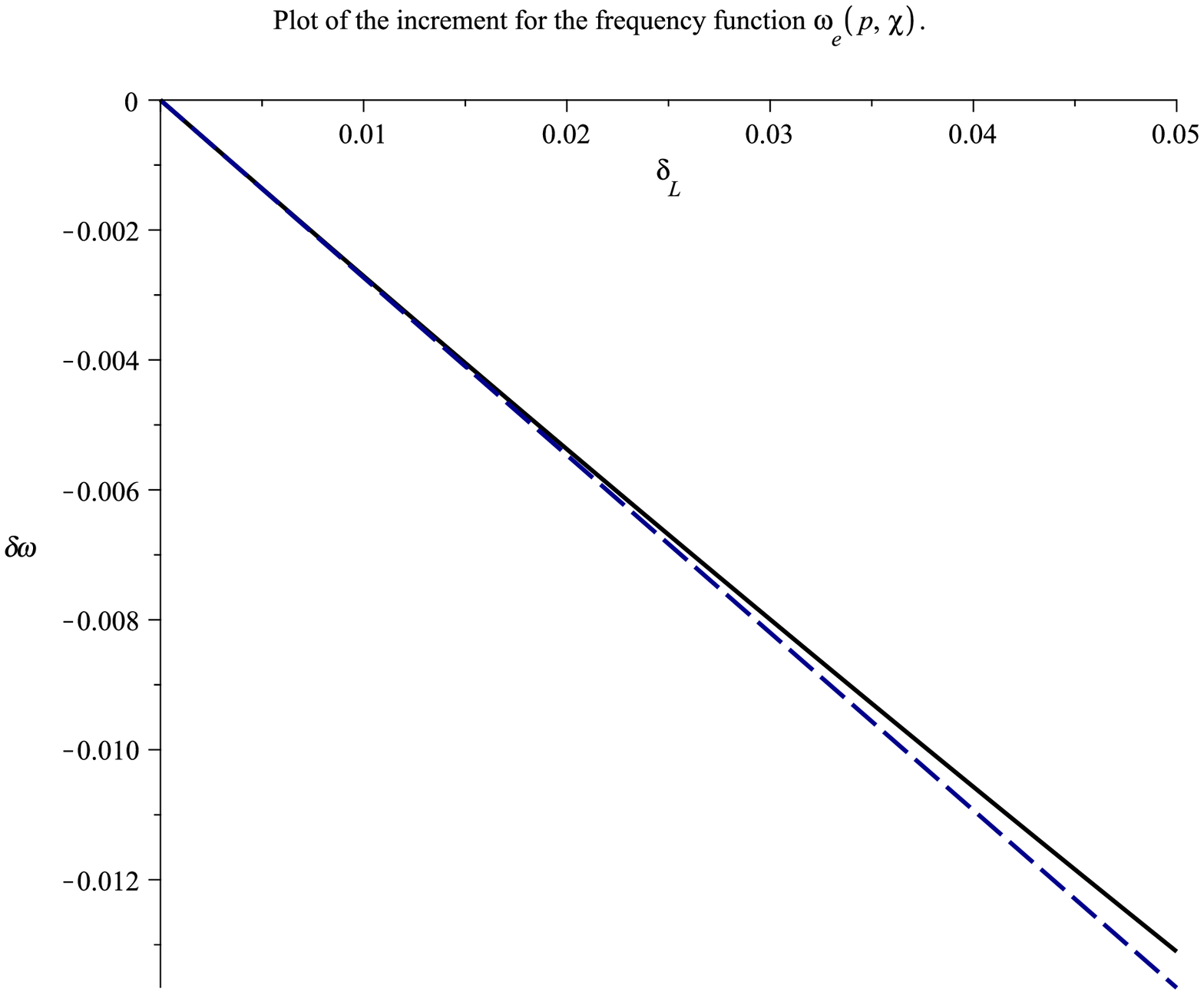}
\par\end{centering}
\centering{}(a)\hspace{7cm}(b)\caption{\label{fig:incr-om-L} Plots of the increment $\check{\omega}_{\mathrm{e}}\left(p^{\prime},\chi{}^{\prime}\right)-\check{\omega}_{\mathrm{e}}\left(p,\chi\right)$
and its first order approximation as in equation (\ref{eq:upS2d})
for $\mathring{v}=1.1$, $K_{\gamma}=1.2$, $C=1.3$, $L=1.4$, $\delta_{C}=0$
and: (a) $0<\delta_{L}<0.3$; (b) $0<\delta_{L}<0.05$.}
\end{figure}

Using equations (\ref{eq:upS1c}) and (\ref{eq:upS2d}) we recover
the value of $\delta_{\check{\omega}}$ based on the measured velocity
shift $\check{u}_{+}-\check{u}_{-}$ as follows
\begin{equation}
\delta_{\check{\omega}}\cong-\frac{p^{2}\delta_{C}+{\it \chi}^{2}\delta_{L}}{2\left({\it \chi}^{2}-p^{3}\right)}=\frac{\left(\check{u}_{+}-\check{u}_{-}\right)^{2}}{4S\check{\omega}_{\mathrm{e}}}+\delta_{\mathrm{w}}=\frac{\left(\check{u}_{+}-\check{u}_{-}\right)^{2}}{4S\check{\omega}_{\mathrm{e}}}+\frac{\check{\omega}_{\mathrm{w}}}{\check{\omega}_{\mathrm{e}}}-1.\label{eq:delCLom1a}
\end{equation}
Consequently, if the distributed capacitance $C$ was used for sensing
and hence $\delta_{L}=0$ we obtain from equations (\ref{eq:delCLom1a})
the value of $\delta_{C}$, that is
\begin{equation}
\delta_{C}\cong\frac{2\left({\it \chi}^{2}-p^{3}\right)}{p^{2}}\left[1-\frac{\check{\omega}_{\mathrm{w}}}{\check{\omega}_{\mathrm{e}}}-\frac{\left(\check{u}_{+}-\check{u}_{-}\right)^{2}}{4S\check{\omega}_{\mathrm{e}}}\right].\label{eq:delCLom1b}
\end{equation}
Similarly, if the distributed inductance $L$ was used for sensing
and hence $\delta_{C}=0$ we obtain from equations (\ref{eq:delCLom1a})
the value of $\delta_{L}$, that is
\begin{equation}
\delta_{L}\cong\frac{2\left({\it \chi}^{2}-p^{3}\right)}{\chi^{2}}\left[1-\frac{\check{\omega}_{\mathrm{w}}}{\check{\omega}_{\mathrm{e}}}-\frac{\left(\check{u}_{+}-\check{u}_{-}\right)^{2}}{4S\check{\omega}_{\mathrm{e}}}\right].\label{eq:delCLom1c}
\end{equation}

\renewcommand{\sectionname}{Appendix}
\counterwithin{section}{part}
\renewcommand{\thesection}{\Alph{section}}
\setcounter{section}{0}
\renewcommand{\theequation}{\Alph{section}.\arabic{equation}}

\section{Fourier transform\label{sec:four}}

Our preferred form of the Fourier transforms as in \cite[7.2, 7.5]{Foll},
\cite[20.2]{ArfWeb}:
\begin{gather}
f\left(t\right)=\int_{-\infty}^{\infty}\hat{f}\left(\omega\right)\mathrm{e}^{-\mathrm{i}\omega t}\,\mathrm{d}\omega,\quad\hat{f}\left(\omega\right)=\frac{1}{2\pi}\int_{-\infty}^{\infty}f\left(t\right)e^{\mathrm{i}\omega t}\,\mathrm{d}t,\label{eq:fourier1a}\\
f\left(z,t\right)=\int_{-\infty}^{\infty}\hat{f}\left(k,\omega\right)\mathrm{e}^{-\mathrm{i}\left(\omega t-kz\right)}\,\mathrm{d}k\mathrm{d}\omega,\label{eq:fourier1b}\\
\hat{f}\left(k,\omega\right)=\frac{1}{\left(2\pi\right)^{2}}\int_{-\infty}^{\infty}f\left(z,t\right)e^{\mathrm{i}\left(\omega t-kz\right)}\,dz\mathrm{d}t.\nonumber 
\end{gather}
This preference was motivated by the fact that the so-defined Fourier
transform of the convolution of two functions has its simplest form.
Namely, the convolution $f\ast g$ of two functions $f$ and $g$
is defined by \cite[7.2, 7.5]{Foll},
\begin{gather}
\left[f\ast g\right]\left(t\right)=\left[g\ast f\right]\left(t\right)=\int_{-\infty}^{\infty}f\left(t-t^{\prime}\right)g\left(t^{\prime}\right)\,\mathrm{d}t^{\prime},\label{eq:fourier2a}\\
\left[f\ast g\right]\left(z,t\right)=\left[g\ast f\right]\left(z,t\right)=\int_{-\infty}^{\infty}f\left(z-z^{\prime},t-t^{\prime}\right)g\left(z^{\prime},t^{\prime}\right)\,\mathrm{d}z^{\prime}\mathrm{d}t^{\prime}.\label{eq:fourier2b}
\end{gather}
 Then its Fourier transform as defined by equations (\ref{eq:fourier1a})
and (\ref{eq:fourier1b}) satisfies the following properties:
\begin{gather}
\widehat{f\ast g}\left(\omega\right)=\hat{f}\left(\omega\right)\hat{g}\left(\omega\right),\label{eq:fourier3a}\\
\widehat{f\ast g}\left(k,\omega\right)=\hat{f}\left(k,\omega\right)\hat{g}\left(k,\omega\right).\label{eq:fourier3b}
\end{gather}

\section{Matrix polynomial and its companion matrix\label{sec:mat-poly}}

An important incentive for considering matrix polynomials is that
they are relevant to the spectral theory of the differential equations
of the order higher than 1, particularly the Euler-Lagrange equations
which are the second-order differential equations in time. We provide
here selected elements of the theory of matrix polynomials following
mostly to \cite[II.7, II.8]{GoLaRo}, \cite[9]{Baum}. General matrix
polynomial eigenvalue problem reads
\begin{equation}
A\left(s\right)x=0,\quad A\left(s\right)=\sum_{j=0}^{\nu}A_{j}s^{j},\quad x\neq0,\label{eq:Aux1a}
\end{equation}
where $s$ is complex number, $A_{k}$ are constant $m\times m$ matrices
and $x\in\mathbb{C}^{m}$ is $m$-dimensional column-vector. We refer
to problem (\ref{eq:Aux1a}) of funding complex-valued $s$ and non-zero
vector $x\in\mathbb{C}^{m}$ as polynomial eigenvalue problem. 

If a pair of a complex $s$ and non-zero vector $x$ solves problem
(\ref{eq:Aux1a}) we refer to $s$ as an \emph{eigenvalue} or as a\emph{
characteristic value} and to $x$ as the corresponding to $s$ \emph{eigenvector}.
Evidently the characteristic values of problem (\ref{eq:Aux1a}) can
be found from polynomial \emph{characteristic equation}
\begin{equation}
\det\left\{ A\left(s\right)\right\} =0.\label{eq:Aux1b}
\end{equation}
We refer to matrix polynomial $A\left(s\right)$ as \emph{regular}
if $\det\left\{ A\left(s\right)\right\} $ is not identically zero.
We denote by $m\left(s_{0}\right)$ the \emph{multiplicity} (called
also \emph{algebraic multiplicity}) of eigenvalue $s_{0}$ as a root
of polynomial $\det\left\{ A\left(s\right)\right\} $. In contrast,
the \emph{geometric multiplicity} of eigenvalue $s_{0}$ is defined
as $\dim\left\{ \ker\left\{ A\left(s_{0}\right)\right\} \right\} $,
where $\ker\left\{ A\right\} $ defined for any square matrix $A$
stands for the subspace of solutions $x$ to equation $Ax=0$. Evidently,
the geometric multiplicity of eigenvalue does not exceed its algebraic
one, see Corollary \ref{cor:dim-ker}. 

It turns out that the matrix polynomial eigenvalue problem (\ref{eq:Aux1a})
can be always recast as the standard ``linear'' eigenvalue problem,
namely
\begin{equation}
\left(s\mathsf{B}-\mathsf{A}\right)\mathsf{x}=0,\label{eq:Aux1c}
\end{equation}
where $m\nu\times m\nu$ matrices $\mathsf{A}$ and $\mathsf{B}$
are defined by
\begin{gather}
\mathsf{B}=\left[\begin{array}{ccccc}
\mathbb{I} & 0 & \cdots & 0 & 0\\
0 & \mathbb{I} & 0 & \cdots & 0\\
0 & 0 & \ddots & \cdots & \vdots\\
\vdots & \vdots & \ddots & \mathbb{I} & 0\\
0 & 0 & \cdots & 0 & A_{\nu}
\end{array}\right],\quad\mathsf{A}=\left[\begin{array}{ccccc}
0 & \mathbb{I} & \cdots & 0 & 0\\
0 & 0 & \mathbb{I} & \cdots & 0\\
0 & 0 & 0 & \cdots & \vdots\\
\vdots & \vdots & \ddots & 0 & \mathbb{I}\\
-A_{0} & -A_{1} & \cdots & -A_{\nu-2} & -A_{\nu-1}
\end{array}\right],\label{eq:CBA1b}
\end{gather}
with $\mathbb{I}$ being $m\times m$ identity matrix. Matrix $\mathsf{A}$,
particularly in monic case, is often referred to as \emph{companion
matrix}. In the case of \emph{monic polynomial} $A\left(\lambda\right)$,
when $A_{\nu}=\mathbb{I}$ is $m\times m$ identity matrix, matrix
$\mathsf{B}=\mathsf{I}$ is $m\nu\times m\nu$ identity matrix. The
reduction of original polynomial problem (\ref{eq:Aux1a}) to an equivalent
linear problem (\ref{eq:Aux1c}) is called \emph{linearization}.

The linearization is not unique, and one way to accomplish is by introducing
the so-called known ``\emph{companion polynomia}l'' which is $m\nu\times m\nu$
matrix
\begin{gather}
\mathsf{C}_{A}\left(s\right)=s\mathsf{B}-\mathsf{A}=\left[\begin{array}{ccccc}
s\mathbb{I} & -\mathbb{I} & \cdots & 0 & 0\\
0 & s\mathbb{I} & -\mathbb{I} & \cdots & 0\\
0 & 0 & \ddots & \cdots & \vdots\\
\vdots & \vdots & \vdots & s\mathbb{I} & -\mathbb{I}\\
A_{0} & A_{1} & \cdots & A_{\nu-2} & sA_{\nu}+A_{\nu-1}
\end{array}\right].\label{eq:CBA1a}
\end{gather}
Notice that in the case of the EL equations the linearization can
be accomplished by the relevant Hamilton equations.

To demonstrate the equivalency between the eigenvalue problems for
$m\nu\times m\nu$ companion polynomial $\mathsf{C}_{A}\left(s\right)$
and the original $m\times m$ matrix polynomial $A\left(s\right)$
we introduce two $m\nu\times m\nu$ matrix polynomials $\mathsf{E}\left(s\right)$
and $\mathsf{F}\left(s\right)$. Namely,
\begin{gather}
\mathsf{E}\left(s\right)=\left[\begin{array}{ccccc}
E_{1}\left(s\right) & E_{2}\left(s\right) & \cdots & E_{\nu-1}\left(s\right) & \mathbb{I}\\
-\mathbb{I} & 0 & 0 & \cdots & 0\\
0 & -\mathbb{I} & \ddots & \cdots & \vdots\\
\vdots & \vdots & \ddots & 0 & 0\\
0 & 0 & \cdots & -\mathbb{I} & 0
\end{array}\right],\label{eq:CBA1c}\\
\det\left\{ \mathsf{E}\left(s\right)\right\} =1,\nonumber 
\end{gather}
where $m\times m$ matrix polynomials $E_{j}\left(s\right)$ are defined
by the following recursive formulas
\begin{gather}
E_{\nu}\left(s\right)=A_{\nu},\quad E_{j-1}\left(s\right)=A_{j-1}+sE_{j}\left(s\right),\quad j=\nu,\ldots,2.\label{eq:CBA1d}
\end{gather}
Matrix polynomial $\mathsf{F}\left(s\right)$ is defined by
\begin{gather}
\mathsf{F}\left(s\right)=\left[\begin{array}{ccccc}
\mathbb{I} & 0 & \cdots & 0 & 0\\
-s\mathbb{I} & \mathbb{I} & 0 & \cdots & 0\\
0 & -s\mathbb{I} & \ddots & \cdots & \vdots\\
\vdots & \vdots & \ddots & \mathbb{I} & 0\\
0 & 0 & \cdots & -s\mathbb{I} & \mathbb{I}
\end{array}\right],\quad\det\left\{ \mathsf{F}\left(s\right)\right\} =1.\label{eq:CBA1e}
\end{gather}
Notice, that both matrix polynomials $\mathsf{E}\left(s\right)$ and
$\mathsf{F}\left(s\right)$ have constant determinants readily implying
that their inverses $\mathsf{E}^{-1}\left(s\right)$ and $\mathsf{F}^{-1}\left(s\right)$
are also matrix polynomials. Then it is straightforward to verify
that
\begin{gather}
\mathsf{E}\left(s\right)\mathsf{C}_{A}\left(s\right)\mathsf{F}^{-1}\left(s\right)=\mathsf{E}\left(s\right)\left(s\mathsf{B}-\mathsf{A}\right)\mathsf{F}^{-1}\left(s\right)=\left[\begin{array}{ccccc}
A\left(s\right) & 0 & \cdots & 0 & 0\\
0 & \mathbb{I} & 0 & \cdots & 0\\
0 & 0 & \ddots & \cdots & \vdots\\
\vdots & \vdots & \ddots & \mathbb{I} & 0\\
0 & 0 & \cdots & 0 & \mathbb{I}
\end{array}\right].\label{eq:CBA1f}
\end{gather}
The identity (\ref{eq:CBA1f}) where matrix polynomials $\mathsf{E}\left(s\right)$
and $\mathsf{F}\left(s\right)$ have constant determinants can be
viewed as the definition of equivalency between matrix polynomial
$A\left(s\right)$ and its companion polynomial $\mathsf{C}_{A}\left(s\right)$. 

Let us take a look at the eigenvalue problem for eigenvalue $s$ and
eigenvector $\mathsf{x}\in\mathbb{C}^{m\nu}$ associated with companion
polynomial $\mathsf{C}_{A}\left(s\right)$, that is
\begin{gather}
\left(s\mathsf{B}-\mathsf{A}\right)\mathsf{x}=0,\quad\mathsf{x}=\left[\begin{array}{c}
x_{0}\\
x_{1}\\
x_{2}\\
\vdots\\
x_{\nu-1}
\end{array}\right]\in\mathbb{C}^{m\nu},\quad x_{j}\in\mathbb{C}^{m},\quad0\leq j\leq\nu-1,\label{eq:CBAx1a}
\end{gather}
where
\begin{equation}
\left(s\mathsf{B}-\mathsf{A}\right)\mathsf{x}=\left[\begin{array}{c}
sx_{0}-x_{1}\\
sx_{1}-x_{2}\\
\vdots\\
sx_{\nu-2}-x_{\nu-1}\\
\sum_{j=0}^{\nu-2}A_{j}x_{j}+\left(sA_{\nu}+A_{\nu-1}\right)x_{\nu-1}
\end{array}\right].\label{eq:CBAx1b}
\end{equation}
With equations (\ref{eq:CBAx1a}) and (\ref{eq:CBAx1b}) in mind we
introduce the following vector polynomial
\begin{equation}
\mathsf{x}_{s}=\left[\begin{array}{c}
x_{0}\\
sx_{0}\\
\vdots\\
s^{\nu-2}x_{0}\\
s^{\nu-1}x_{0}
\end{array}\right],\quad x_{0}\in\mathbb{C}^{m}.\label{eq:CBAx1c}
\end{equation}
Not accidentally, the components of the vector $\mathsf{x}_{s}$ in
its representation (\ref{eq:CBAx1c}) are in evident relation with
the derivatives $\partial_{t}^{j}\left(x_{0}\mathrm{e}^{st}\right)=s^{j}x_{0}\mathrm{e}^{st}$.
That is just another sign of the intimate relations between the matrix
polynomial theory and the theory of systems of ordinary differential
equations \cite[III.4]{Hale}.
\begin{thm}[eigenvectors]
\label{thm:matpol-eigvec} Let $A\left(s\right)$ as in equations
(\ref{eq:Aux1a}) be regular, that $\det\left\{ A\left(s\right)\right\} $
is not identically zero, and let $m\nu\times m\nu$ matrices $\mathsf{A}$
and $\mathsf{B}$ be defined by equations (\ref{eq:Aux1b}). Then
the following identities hold
\begin{equation}
\left(s\mathsf{B}-\mathsf{A}\right)\mathsf{x}_{s}=\left[\begin{array}{c}
0\\
0\\
\vdots\\
0\\
A\left(s\right)x_{0}
\end{array}\right],\;\mathsf{x}_{s}=\left[\begin{array}{c}
x_{0}\\
sx_{0}\\
\vdots\\
s^{\nu-2}x_{0}\\
s^{\nu-1}x_{0}
\end{array}\right],\label{eq:CBAx1d}
\end{equation}
\begin{gather}
\det\left\{ A\left(s\right)\right\} =\det\left\{ s\mathsf{B}-\mathsf{A}\right\} ,\quad\det\left\{ \mathsf{B}\right\} =\det\left\{ A_{\nu}\right\} ,\label{eq:CBAx1g}
\end{gather}
where $\det\left\{ A\left(s\right)\right\} =\det\left\{ s\mathsf{B}-\mathsf{A}\right\} $
is a polynomial of the degree $m\nu$ if $\det\left\{ \mathsf{B}\right\} =\det\left\{ A_{\nu}\right\} \neq0$.
There is one-to-one correspondence between solutions of equations
$A\left(s\right)x=0$ and $\left(s\mathsf{B}-\mathsf{A}\right)\mathsf{x}=0$.
Namely, a pair $s,\:\mathsf{x}$ solves eigenvalue problem $\left(s\mathsf{B}-\mathsf{A}\right)\mathsf{x}=0$
if and only if the following equalities hold
\begin{gather}
\mathsf{x}=\mathsf{x}_{s}=\left[\begin{array}{c}
x_{0}\\
sx_{0}\\
\vdots\\
s^{\nu-2}x_{0}\\
s^{\nu-1}x_{0}
\end{array}\right],\quad A\left(s\right)x_{0}=0,\quad x_{0}\neq0;\quad\det\left\{ A\left(s\right)\right\} =0.\label{eq:CBAx1e}
\end{gather}
\end{thm}

\begin{proof}
Polynomial vector identity (\ref{eq:CBAx1d}) readily follows from
equations (\ref{eq:CBAx1b}) and (\ref{eq:CBAx1c}). Identities (\ref{eq:CBAx1g})
for the determinants follow straightforwardly from equations (\ref{eq:CBAx1c}),
(\ref{eq:CBAx1e}) and (\ref{eq:CBA1f}). If $\det\left\{ \mathsf{B}\right\} =\det\left\{ A_{\nu}\right\} \neq0$
then the degree of the polynomial $\det\left\{ s\mathsf{B}-\mathsf{A}\right\} $
has to be $m\nu$ since $\mathsf{A}$ and $\mathsf{B}$ are $m\nu\times m\nu$
matrices.

Suppose that equations (\ref{eq:CBAx1e}) hold. Then combining them
with proven identity (\ref{eq:CBAx1d}) we get $\left(s\mathsf{B}-\mathsf{A}\right)\mathsf{x}_{s}=0$
proving that expressions (\ref{eq:CBAx1e}) define an eigenvalue $s$
and an eigenvector $\mathsf{x}=\mathsf{x}_{s}$.

Suppose now that $\left(s\mathsf{B}-\mathsf{A}\right)\mathsf{x}=0$
where $\mathsf{x}\neq0$. Combing that with equations (\ref{eq:CBAx1b})
we obtain
\begin{gather}
x_{1}=sx_{0},\quad x_{2}=sx_{1}=s^{2}x_{0},\cdots,\quad x_{\nu-1}=s^{\nu-1}x_{0},\label{eq:CBAx2a}
\end{gather}
implying that
\begin{equation}
\mathsf{x}=\mathsf{x}_{s}=\left[\begin{array}{c}
x_{0}\\
sx_{0}\\
\vdots\\
s^{\nu-2}x_{0}\\
s^{\nu-1}x_{0}
\end{array}\right],\quad x_{0}\neq0,\label{eq:CBAx2b}
\end{equation}
 and 
\begin{equation}
\sum_{j=0}^{\nu-2}A_{j}x_{j}+\left(sA_{\nu}+A_{\nu-1}\right)x_{\nu-1}=A\left(s\right)x_{0}.\label{eq:CBAx2c}
\end{equation}
Using equations (\ref{eq:CBAx2b}) and identity (\ref{eq:CBAx1d})
we obtain
\begin{equation}
0=\left(s\mathsf{B}-\mathsf{A}\right)\mathsf{x}=\left(s\mathsf{B}-\mathsf{A}\right)\mathsf{x}_{s}=\left[\begin{array}{c}
0\\
0\\
\vdots\\
0\\
A\left(s\right)x_{0}
\end{array}\right].\label{eq:CBAx2d}
\end{equation}
 Equations (\ref{eq:CBAx2d}) readily imply $A\left(s\right)x_{0}=0$
and $\det\left\{ A\left(s\right)\right\} =0$ since $x_{0}\neq0$.
That completes the proof.
\end{proof}
\begin{rem}[characteristic polynomial degree]
\label{rem:char-pol-deg} Notice that according to Theorem \ref{thm:matpol-eigvec}
the characteristic polynomial $\det\left\{ A\left(s\right)\right\} $
for $m\times m$ matrix polynomial $A\left(s\right)$ has the degree
$m\nu$, whereas in linear case $s\mathbb{I}-A_{0}$ for $m\times m$
identity matrix $\mathbb{I}$ and $m\times m$ matrix $A_{0}$ the
characteristic polynomial $\det\left\{ s\mathbb{I}-A_{0}\right\} $
is of the degree $m$. This can be explained by observing that in
the non-linear case of $m\times m$ matrix polynomial $A\left(s\right)$
we are dealing effectively with many more $m\times m$ matrices $A$
than just a single matrix $A_{0}$.
\end{rem}

Another problem of our particular interest related to the theory of
matrix polynomials is eigenvalues and eigenvectors degeneracy and
consequently the existence of non-trivial Jordan blocks, that is Jordan
blocks of dimensions higher or equal to 2. The general theory addresses
this problem by introducing so-called ``Jordan chains'' which are
intimately related to the theory of system of differential equations
expressed as $A\left(\partial_{t}\right)x\left(t\right)=0$ and their
solutions of the form $x\left(t\right)=p\left(t\right)e^{st}$ where
$p\left(t\right)$ is a vector polynomial, see \cite[I, II]{GoLaRo},
\cite[9]{Baum}. Avoiding the details of Jordan chains developments
we simply notice that an important to us point of Theorem \ref{thm:matpol-eigvec}
is that there is one-to-one correspondence between solutions of equations
$A\left(s\right)x=0$ and $\left(s\mathsf{B}-\mathsf{A}\right)\mathsf{x}=0$,
and it has the following immediate implication.
\begin{cor}[equality of the dimensions of eigenspaces]
\label{cor:dim-ker} Under the conditions of Theorem \ref{thm:matpol-eigvec}
for any eigenvalue $s_{0}$, that is $\det\left\{ A\left(s_{0}\right)\right\} =0$,
we have
\begin{equation}
\dim\left\{ \ker\left\{ s_{0}\mathsf{B}-\mathsf{A}\right\} \right\} =\dim\left\{ \ker\left\{ A\left(s_{0}\right)\right\} \right\} .\label{eq:CBAx2e}
\end{equation}
In other words, the geometric multiplicities of the eigenvalue $s_{0}$
associated with matrices $A\left(s_{0}\right)$ and $s_{0}\mathsf{B}-\mathsf{A}$
are equal. In view of identity (\ref{eq:CBAx2e}) the following inequality
holds for the (algebraic) multiplicity $m\left(s_{0}\right)$
\begin{equation}
m\left(s_{0}\right)\geq\dim\left\{ \ker\left\{ A\left(s_{0}\right)\right\} \right\} .\label{eq:CBAx2f}
\end{equation}
\end{cor}

The next statement shows that if the geometric multiplicity of an
eigenvalue is strictly less than its algebraic one than there exist
non-trivial Jordan blocks, that is Jordan blocks of dimensions higher
or equal to 2.
\begin{thm}[non-trivial Jordan block]
\label{thm:Jord-block} Assuming notations introduced in Theorem
\ref{thm:matpol-eigvec} let us suppose that the multiplicity $m\left(s_{0}\right)$
of eigenvalue $s_{0}$ satisfies
\begin{equation}
m\left(s_{0}\right)>\dim\left\{ \ker\left\{ A\left(s_{0}\right)\right\} \right\} .\label{eq:CBAAx3a}
\end{equation}
Then the Jordan canonical form of companion polynomial $\mathsf{C}_{A}\left(s\right)=s\mathsf{B}-\mathsf{A}$
has a least one nontrivial Jordan block of the dimension exceeding
2.

In particular, if 
\begin{equation}
\dim\left\{ \ker\left\{ s_{0}\mathsf{B}-\mathsf{A}\right\} \right\} =\dim\left\{ \ker\left\{ A\left(s_{0}\right)\right\} \right\} =1,\label{eq:CBAAx3b}
\end{equation}
and $m\left(s_{0}\right)\geq2$ then the Jordan canonical form of
companion polynomial $\mathsf{C}_{A}\left(s\right)=s\mathsf{B}-\mathsf{A}$
has exactly one Jordan block associated with eigenvalue $s_{0}$ and
its dimension is $m\left(s_{0}\right)$.
\end{thm}

The proof of Theorem \ref{thm:Jord-block} follows straightforwardly
from the definition of the Jordan canonical form and its basic properties.
Notice that if equations (\ref{eq:CBAAx3b}) hold that implies that
the eigenvalue $0$ is cyclic (nonderogatory) for matrix $A\left(s_{0}\right)$
and eigenvalue $s_{0}$ is cyclic (nonderogatory) for matrix $\mathsf{B}^{-1}\mathsf{A}$
provided $\mathsf{B}^{-1}$ exists. We remind that an eigenvalue is
called \emph{cyclic (nonderogatory)} if its geometric multiplicity
is 1. A square matrix is called \emph{cyclic (nonderogatory)} if all
its eigenvalues are cyclic \cite[5.5]{BernM}.

\section{Notations\label{sec:notation}}
\begin{itemize}
\item $\mathbb{C}$ is a set of complex number.
\item $\bar{s}$ is complex-conjugate to complex number $s$
\item $\mathbb{C}^{n}$ is a set of $n$ dimensional column vectors with
complex complex-valued entries.
\item $\mathbb{C}^{n\times m}$ is a set of $n\times m$ matrices with complex-valued
entries.
\item $\mathbb{R}^{n\times m}$ is a set of $n\times m$ matrices with real-valued
entries.
\item $\mathrm{dim}\,\left(W\right)$ is the dimension of the vector space
$W$.
\item $\mathrm{ker}\,\left(A\right)$ is the kernel of matrix $A$, that
is the vector space of vector $x$ such that $Ax=0$.
\item $\det\left\{ A\right\} $ is the determinant of matrix $A$.
\item $\chi_{A}\left(s\right)=\det\left\{ s\mathbb{I}_{\nu}-A\right\} $
is the characteristic polynomial of a $\nu\times\nu$ matrix $A$.
\item $\mathbb{I}_{\nu}$ is $\nu\times\nu$ identity matrix.
\item $M^{\mathrm{T}}$ is a matrix transposed to matrix $M$.
\item EL stands for the Euler-Lagrange (equations).
\end{itemize}
\textbf{\vspace{0.1cm}
}

\textbf{Data Availability:} The data that supports the findings of
this study are available within the article.\textbf{\vspace{0.2cm}
}

\textbf{Acknowledgment:} This research was supported by AFOSR grant
\# FA9550-19-1-0103.

\end{document}